\documentclass[aoas]{imsart}

\RequirePackage{amsthm,amsmath,amsfonts,amssymb}
\RequirePackage[authoryear]{natbib}

\startlocaldefs
\usepackage{array}
\usepackage{amsmath}
\usepackage{graphicx}
\usepackage{enumerate}
\usepackage{url} 
\usepackage{adjustbox}
\usepackage{bigints}
\usepackage[ruled, vlined, linesnumbered]{algorithm2e}
\usepackage{babel}
\usepackage{amsthm}
\usepackage{amssymb}
\usepackage{color}
\usepackage{float}
\usepackage{mathtools}
\usepackage{subcaption}

\providecommand{\tabularnewline}{\\}
\floatstyle{ruled}
\newfloat{algorithm}{tbp}{loa}
\providecommand{\algorithmname}{Algorithm}
\floatname{algorithm}{\protect\algorithmname}

\numberwithin{equation}{section}
\numberwithin{figure}{section}

\newtheorem{lemma}{Lemma}[section]

\theoremstyle{plain}
\newtheorem{thm}{\protect\theoremname}
\theoremstyle{remark}
\newtheorem{rem}[thm]{\protect\remarkname}
\theoremstyle{plain}
\newtheorem{prop}[thm]{\protect\propositionname}
\theoremstyle{plain}
\newtheorem{lem}[thm]{\protect\lemmaname}
\theoremstyle{definition}
\newtheorem{defn}[thm]{\protect\definitionname}

\newtheorem*{assumption*}{\assumptionnumber}
\providecommand{\assumptionnumber}{}


\newtheorem*{theorem*}{\theoremnumber}
\providecommand{\theoremnumber}{}


\SetKwInput{KwInput}{Input}
\SetKwInput{KwOutput}{Output}
\let\oldnl\nl
\newcommand{\nonl}{\renewcommand{\nl}{\let\nl\oldnl}}

\providecommand{\definitionname}{Definition}
\providecommand{\lemmaname}{Lemma}
\providecommand{\propositionname}{Proposition}
\providecommand{\remarkname}{Remark}
\providecommand{\theoremname}{Theorem}

\endlocaldefs

\begin{document}

\begin{frontmatter}

\title{Confidence Intervals for Rate Estimation with Importance Sampling in Autonomous Vehicle Evaluation}
\runtitle{CIs for rate estimation}

\begin{aug}

\author[]{\fnms{Aiyou}~\snm{Chen}}\ead[label=e1]{aiyouchen@waymo.com},
\author[]{\fnms{Ruixuan Rachel}~\snm{Zhou}}\ead[label=e2]{rachelzhou@waymo.com},
\author[]{\fnms{Joseph J.}~\snm{Lee}}\ead[label=e3]{josephjlee@waymo.com},
\author[]{\fnms{Nicholas}~\snm{Chamandy}}\ead[label=e4]{chamandy@waymo
.com}\\
\and
\author[]{\fnms{Henning}~\snm{Hohnhold}}\ead[label=e5]{henningh@waymo.com}
\address[]{Waymo LLC \printead[presep={ ,\ }]{e1,e2,e3,e4,e5}}

\end{aug}

\begin{abstract}
Accounting for both rare events and complex sampling presents challenges when quantifying uncertainty for rate estimation in autonomous vehicle performance evaluation.
In this paper, we introduce a statistical formulation of this
problem and develop a unified compound Poisson model framework for
unbiased rate estimation through the Horvitz-Thompson estimator. Though
asymptotic theory for the model is available, the inference of confidence intervals (CIs) in the presence of
rare events requires new investigation. We also advocate for a new monotonicity criterion
for rate CIs—summing the rates of disjoint types of events should produce not only a
higher point estimate but also higher confidence bounds than for the individual
rates—that facilitates
interpretability in real applications. We propose a novel \textit{exponential
bootstrap} (EB) method for CI construction based on a fiducial argument; it satisfies the monotonicity property, while novel extensions of some existing
methods do not. Comprehensive numerical studies show that EB performs
well for a wide range of settings relevant to our applications. 
Fast implementation of EB based on saddlepoint approximation is also developed, which may be of independent interest.

\end{abstract}


\end{frontmatter}


\section{Introduction}
\label{sec:intro}

Assessing driving performance is a central element of ensuring a safe and scalable deployment in the autonomous vehicle (AV) industry. As AVs approach and even surpass some human driving capabilities, the vast majority of their driving may become routine and uninteresting. In turn, assessing their performance often translates to estimating the frequency of rare events. Events of highest interest tend to be the rarest, e.g., collisions with potentially high-severity injuries only happen once every several million human-driven miles  \citep{kalra2016driving, kusano2024comparison}. Even events of interest that do not pose an immediate safety risk, such as those in which a vehicle contributes to unnecessary traffic congestion, can be rare enough to invalidate common statistical methods, be difficult to detect in an automated fashion, and pose other unique measurement challenges (e.g., how does one assess the vehicle’s causal impact on congestion?). Accurate estimation of AV performance thus requires both large-scale data collection and the ability to efficiently extract a very sparse signal from those data. In short, AV performance measurement is a ``needle-in-a-haystack'' problem. See  \cite{kusano2024comparison,chen2024dynamic,di2024comparative,di2024latest} for some recent data regarding AV performance as well as comparisons with human benchmarks.

The software that powers an AV system undergoes frequent iteration and improvement. Evaluating the performance of a new version of self-driving software accurately and quickly can be challenging, as collecting a large volume of on-road test driving data is slow. Suppose that we want to evaluate today's software version and that it takes $k$ days to collect a sufficiently large volume of on-road test data for the desired level of statistical precision. If we find after one round of data collection that there is need for improvement, we must first implement the required software changes, and then wait an additional $k$ days to evaluate the system again. Now, consider the fact that there are actually many aspects of software that may interact with each other, being developed in parallel—e.g., decision logic to safely change lanes, as well as logic to select the optimal route to the destination. This cycle of waiting for on-road test driving data, making code changes, and waiting again to validate those changes makes AV software development purely using on-road test driving infeasible on any reasonable time scale. In other words, the “feedback period” $k$ is typically much longer than the time needed for engineers to implement software improvements. In addition, to ensure safety,
on-road testing of brand new software involves the in-car presence of a human safety operator \citep{webb.etc:2020}.

Performance evaluation using simulation-based ``virtual driving'' is therefore
an important complement to on-road testing in supporting greater
speed and scale for AV development.\footnote{For example, see https://waymo.com/blog/2021/07/simulation-city/} For the purposes of this paper, “simulation” refers to virtual cars “driving” via a computer simulation engine through situations logged from previous real-world driving, in order to measure the software’s performance. Note that in a given simulation, the virtual AV can be configured to drive according to a specific software version, in order to mimic real-world behavior as closely as possible. Simulation provides a safer and faster evaluation platform than on-road testing: safer because we can allow simulated situations to play out regardless of their outcome, without needing intervention from a human safety operator, and faster because simulated cars can drive in the virtual world orders of magnitude faster than real time.

Simulation can also provide signal amplification. Most driving events of interest occur as a result of the accumulation of subtle but important factors. For example, a small perturbation of the AV’s position or speed when a traffic light turns yellow can make a previously straightforward “stop or go” decision challenging. Or, slightly different behaviors from cars in adjacent traffic stacks can influence whether changing lanes is a good decision. Simulating a single driving log under a variety of different conditions can aid the discovery of sub-optimal outcomes and amplify statistical signals. Large amounts of simulated data with amplified signals provide a great first step to work with—see \cite{jiangscenediffuser,mahjourian2024unigen} and references therein.

This large volume of simulated data must undergo highly efficient importance sampling\footnote{For example, see https://artowen.su.domains/mc/Ch-var-is.pdf} before it can be used to construct rate estimators, as described in detail in the next section. Obtaining an unbiased point estimate that accounts for importance sampling is relatively simple, e.g., using a Horvitz-Thompson estimator (\cite{horvitz1952generalization}). But the real-world implications of AV deployment and the sparsity of events together highlight the importance of quantifying an estimate's uncertainty. When evaluating a highly mature system, we often observe point estimates near zero; such estimates are of little use to decision makers, or worse, may paint a misleading picture of performance when variance is large. Understanding the associated uncertainty---e.g., via a confidence interval (CI)---helps strengthen our confidence in a smooth real-world deployment, without unforeseen negative impacts on the communities in which we operate; it also helps design the next iteration of driving software. Developing  a reliable CI in this context is not only important but challenging, and must account for rare events and complex sampling among other things.

In this paper we introduce the problem of statistical inference for rate estimation with importance sampling in AV evaluation based on large-scale simulations and human review, a sequel to the sampling problem discussed in \cite{terres2023behavioral}. We focus on producing confidence intervals and propose a novel exponential bootstrap method, extending the existing Gamma methods (e.g., \cite{fay1997confidence}) and satisfying a statistical monotonicity property that is desirable for interpretability but has not been studied before. The rest of the paper is organized as follows: we first describe the application, data collection procedure, and resulting data set in greater detail in Section 2; we derive a unified statistical model in Section 3; some related work is reviewed in Section 4; the development of the CI methodology and extensive numerical studies are presented in Section 5 and 6, respectively; some real-world data analysis and statistical findings are reported in Section 7; finally, Section 8 concludes with some discussions. Technical derivations and a fast algorithm based on the saddlepoint approximation are presented in the appendix.

\section{Rate estimation: problem and data specification}

To ground the problem description, suppose we wish to identify events where the AV contributes to traffic congestion (``events of interest'') and to estimate the corresponding overall event rate per million driven miles, and the rates of occurrence at different severity levels, e.g., causing a traffic delay having duration within some pre-specified disjoint intervals (less than 1 minute, 1 to 5 minutes, 5+ minutes, etc.). The starting point for this problem is a massive dataset of real, logged on-road driving miles—this could encompass several years of driving data from multiple versions of the self-driving software.
From this large corpus of data, a much smaller (but still large) set of driving miles is selected for simulation with the new software, as described in Section 1. For a sense of scale, the original data set of on-road driving may contain 10s of millions of miles, from which we choose to simulate 10s of thousands to millions of miles, depending on the application.

The next step in rate estimation is to identify events of interest from among the simulated driving miles. Somewhat counter-intuitively, this process is in general not straightforward. Some form of human review is often employed, the need for which stems from two factors. First, the logic that determines a particular event classification may be subtle and complex. In the present application, identifying traffic congestion events of interest may require humans to assess whether the congestion exists exogenously or is truly {\em caused} by an AV's actions. Second, the behavior of other actors (e.g., vehicles, pedestrians) and how they interact with the AV in the simulated scene may need realism assessment, and the event is deemed a false positive if the simulation was unrealistic. For example, a scene in which a real-world actor may have simply navigated around the stopped AV—whereas a simulated vehicle remained stuck behind it—may be deemed unrealistic. Conducting this human assessment is relatively slow compared to simulation, though still much faster than on-road test driving.

Much like selecting a subset of driving logs for simulation, selecting simulated events for human review is a technical challenge. In both stages, probabilistic sampling is employed to reduce the pool of potentially interesting data to satisfy resource constraints, cutting down the data by orders of magnitude while retaining most of the signal of interest. We use importance sampling, and leverage machine learning models designed to favor events with a higher likelihood of being a true positive event of interest. This importance sampling approach allows for unbiased rate estimation while making more efficient use of resources, especially when events are rare. Simple uniform sampling, by contrast, would generate few or even zero events of interest in some applications. These machine learning models are typically tuned to a lower-precision, higher-recall operating point to avoid missing potentially impactful events---therefore satisfying both the rates estimation use case of this paper, and our desire to discover novel types of rare events \citep{sinha2025rate}. From the statistician's perspective, the models used as inputs to importance sampling probabilities can be treated as “black boxes”, and our estimation procedures should be robust to their precise structure and feature space. 

\subsection{Multi-stage importance sampling}
This sampling design plays a large role in efficiently estimating event frequencies, and therefore warrants further elaboration. The data collection process can be described precisely in three major steps, as follows. 

\begin{enumerate}[(1)]
\item Draw a large sample of data points (each data point is a short contiguous block of driving—referred to hereafter as a {\it run segment})\footnote{A brief description of run segments is available in \cite{webb.etc:2020}; see example data at https://waymo.com/open.} by importance sampling from logs
saved from on-road testing over a total of $m$ million miles, favoring interesting driving contexts (e.g. segments on busy roads with other road users in close proximity to the AV); 
\item Simulate each run segment from Step 1 with the AV software to be evaluated, and automatically identify candidate
events of interest (e.g., situations where other vehicles near the AV are driving more slowly than expected); 
\item Draw a sample set of the candidate events obtained from Step 2
by importance sampling; pass these segments to human review, which identifies whether or not
a candidate event is a true positive (meets the criteria for an event of interest, e.g., the AV made congestion around it worse, causing a measurable delay for other road users). 
\end{enumerate}
The sampling probabilities are determined by trained machine learning
models, which predict how likely a segment is to generate a candidate event
from simulation (Step 1), and how likely a candidate event is to be a true
positive (Step 3). To minimize technical complexity, the implementation
of importance sampling is based on Poisson sampling (\cite{mccormick1937sampling}); that is, whether a run segment is sampled for simulation or whether a simulated event is sampled for human review is determined by an \textit{independent} Bernoulli trial, which can be performed in parallel to enable large-scale tasks.
\subsection{Estimand}
We are interested in estimating the expected number of incidents (true positive events of interest) that the AV may encounter per million miles (IPMM), which we denote by $\theta$. As the original set of driving logs represents the driving mile population of interest,
$\theta$ could in theory be estimated by the total number of human-confirmable incidents in this set, divided by $m$; but of course this is only available if all logged
run segments are simulated and all potential candidate events are checked by
human review---which is unrealistic for the reasons discussed previously. While point estimation for $\theta$ is straightforward with standard methods, as formulated in the next section, uncertainty estimation is more subtle. An interval estimator for $\theta$ must not only properly quantify sampling uncertainty (itself challenging due to event rarity and imperfect machine learning models), but must also be interpretable to business stakeholders. In particular, point estimates and CIs must be meaningful both in isolation and when comparing across different levels of aggregation. For example, the lower bound estimate of the rate of 5+ minute congestion events should not exceed the lower bound estimate of the rate of 1+ minute congestion events. This property, which we formalize in Section 5, is not trivially satisfied.

\subsection{Final data set}
The data set available for analysis consists of sampling probabilities for simulating run segments from Step 1, sampling probabilities for reviewing simulated events from Step 3, and the review outcomes. (The data are restricted to run segments for which either simulation or human review occurred.) These variables are formally introduced in the next section for model formulation and statistical inference. Note that the human review outcome variable may consist of multiple columns which are binary or even continuous. For the congestion application, we would expect to have one column with the reviewer's binary assessment of whether or not the candidate event was a true positive AV-caused congestion event, and a second column indicating the duration of the delay. (Typically, the duration variable would be further decomposed into binary variables for each disjoint duration interval.) 
Without loss of generality, hereafter we take $m=1$ for notational simplicity. While our method can be applied more broadly, to minimize technical complexity, we furthermore assume that each run segment may generate at most
one incident of interest. 

\section{The model formulation for rate estimation}

For uncertainty quantification, it is often useful to posit a rigorous statistical model for the data-generating process.
Next we show that the multi-stage importance sampling data collection
procedure as described above can be formulated with a compound Poisson
model framework under some mild assumptions.
Let $N$ be the total number of logged run segments to sample from.

Let $V_{i}$ be the feature, often in high dimensions, associated with the
$i$-th segment which determines the sampling probabilities for both
simulation and human review, $i=1,\cdots,N$. To be precise, the simulation
sampling probability relies on the segment feature only, while the
sampling probability for human review relies on the simulation result,
which depends on not only the feature associated with the segment,
but also the simulator and the software version. Since the simulator and
software are pre-specified, it is conceptually simpler to denote $V_{i}$ as
the feature associated with the segment only.

Let $s(V_{i})$ be the probability that the $i$-th segment
will be sampled for \textit{simulation} and let $h(V_{i})$ be the probability that it will be sampled
for \textit{human} review. Here $h(V_{i})=0$ if the simulation on the $i$th
segment does not generate any candidate event, i.e. there is no need for
human review. We also assume that candidate events will capture
all true positives.\footnote{To relax this assumption, one may require $h(V_i)>0$ even if the simulation on $i$ does not generate any candidate event.} Both $s$ and $h$ are specified by pre-trained machine learning models.

Let $I_{i}\sim Bernoulli(s(V_{i}))$ indicate whether the $i$-th
segment is sampled for simulation and $J_{i}\sim Bernoulli(h(V_{i}))$
indicate whether it is sampled for human review. Let $B_i = I_i J_i$, then
$B_i \sim Bernoulli(p(V_i))$, where $p(v)=s(v)h(v)$ is the overall probability that a segment with feature $v$ will go through both simulation and human review.

Let $Y_{i}$ indicate whether the $i$th segment would generate
a true positive event at the end of the human review process. 
Let $r(V_{i})=P(Y_{i}=1|V_{i})$ be the
true positive probability. Under this model, $\sum_{i=1}^{N}r(V_{i})$ would
be a natural estimator of the rate $\theta$; however, $r$ is unknown. Instead, we adopt the Horvitz-Thompson
estimator 
denoted as $\hat{\theta}$, which is unbiased, i.e. $E(\hat\theta)=\theta$, and
can be formally written as
\begin{align}
\hat{\theta} & =\sum_{i=1}^{N}W_i\label{eq:theta-hat}
\end{align}
where \begin{align}
W_{i} & =\frac{B_{i}}{p(V_{i})}\times Y_{i}.\label{eq:Wi}
\end{align}

For convenience,
we take $0/0$ as $0$ to incorporate the case $h(V_{i})=0$ which
implies $Y_{i}=0$ and $J_{i}=0$. 
Note that given $V_{i}$, $W_{i}$ takes the value $0$ with probability
$1-p(V_{i})r(V_{i})$ and the value $1/p(V_{i})$
with probability $p(V_{i})r(V_{i})$.

\begin{rem}
\label{rem:h(X)}$h(V_{i})$ and $J_{i}$ are observable only if the
$i$-th segment has been simulated, i.e. $I_{i}=1$, while $Y_{i}$
is observable only if $I_{i}=1$ and $J_{i}=1$.
\end{rem}

\begin{prop}
\label{prop:CP}Assume that (1) $N$ follows a Poisson distribution with mean $\lambda$,
and (2) $V_{i}$ are i.i.d. from the feature population, independent
of $N$. Then
\begin{align}
\hat{\theta} \equiv\sum_{i:W_{i}>0}W_{i}\label{eq:CP}
\end{align}
follows a compound Poisson (CP) distribution. 
\end{prop}

Here $\theta$ is the parameter of interest, and $\lambda$ and the distribution of $W_1$ are the nuisance parameters \citep{bickel2015mathematical}.

The proof of Proposition \ref{prop:CP} follows from the fact that $W_{i}$ are i.i.d based on
Assumption 2. Both Assumption 1 and 2 are reasonable when the
run segments are properly defined so that true positive incidents
are independent of each other. How to properly design the run segments to minimize spatial-temporal dependencies is beyond the scope of the current paper but may be discussed elsewhere. Intuitively, event rarity further helps make the independence assumption a good approximation. 
The collection of $\{W_{i}:W_{i}>0\}$
is also referred to as weights later. Note that the Poisson assumption (1) is not essential--when it is violated, under mild conditions $\hat\theta$ would follow an approximate (instead of exact) compound Poisson distribution (see \cite{vcekanavivcius2024compound} and references therein).

Proposition \ref{prop:CP} can be generalized when the response
$Y_{i}$ is not binary but continuous, for example, measuring the duration of a traffic congestion event contributed to by the AV. We assume binary outcomes for the remainder of this paper.

\section{Related work}

The compound Poisson model is a classical statistical model
(\cite{feller2008introduction}) with applications in many areas such
as public health, insurance and astronomy (see \cite{vcekanavivcius2024compound} and references therein).
If the distribution of $W$ can be described by a few parameters,
the model is a parametric model, otherwise it can be viewed as a semi-parametric
model \citep{bickel1993efficient}. Traditional procedures for CI construction include bootstrap
(\cite{efron1992introduction}), the delta method (see \cite{bickel2015mathematical}
for parametric models and \cite{bickel1993efficient} for semi-parametric
models) and empirical likelihood (\cite{owen2001empirical}), which
all rely on large sample approximation. Along these
lines, when analyzing the CP model,
\cite{kegler2007applying} proposed applying a normal approximation in the
logarithmic scale (rather than the standard normal approximation to \eqref{eq:CP}) in order to avoid the potential negative bound,
while \cite{li2012empirical} considered the empirical likelihood CI
by conditioning on the number of observed contributions, i.e. $\sum_{i=1}^{N}I(W_{i}>0)$.
The most interesting approach for our application, due to its simplicity and flexibility, is the Poisson bootstrap (PB) for CP
proposed by \cite{bohm2014statistics};
 see \cite{chamandy2012estimating} and references therein which derived the same Poisson bootstrap but as a variation of the standard multinomial
bootstrap. We also note that PB may be viewed as a parametric bootstrap when the
sampling probabilities are discrete, see details later in Remark \ref{rem:PB-parametric}. Unlike
these methods, which only make use of $\{W_{i}:W_{i}>0\}$, our
method takes into account the sampling probabilities behind all observations.

To handle a small sample setting with the existence of nuisance parameters,
various methods have been developed: for example, the Buehler method \citep{kabaila2006improved}, the repro samples method \citep{xie2024repro} and
the unified method based on likelihood ratio \citep{sen2009unified},
provide warranted coverage, while the generalized fiducial argument
(see \cite{hannig2016generalized}) and the inferential
method \citep{martin2015marginal} provide approximate coverage. These methods are developed for parametric
models and thus do not directly apply to our case since it is hard
to come up with a generic parametric model for $W_{i}$ unless they are discrete. Other directions include synthesizing the inference from multiple studies, e.g. \cite{liu2014exact}, or multiple models, e.g. \cite{agarwal2025pcs} but do not directly apply here either.

When the sampling probabilities are discrete, our model simplifies to a weighted sum of independent Poisson model, for which the problem has been studied extensively and 
the most popular CI methods are
called the Gamma methods (see \cite{fay1997confidence,tiwari2006efficient,fay2017confidence} and references therein). We discovered that these methods however do not have a desired monotonicity property (defined below) that we advocate for in this paper, which is important in the applied setting as it allows practitioners to  communicate a consistent story to business decision-makers.

\section{Statistical inference}

We first consider the case where the sampling probabilities are discrete,
and show that $\hat{\theta}$ reduces to a weighted sum of independent
Poisson model and derive a novel variation of the original Gamma method (\cite{fay1997confidence}),
called weighted Gamma, based on the fiducial argument. Then a natural extension to continuous sampling
probabilities leads to a more general algorithm for the CI construction,
called Exponential bootstrap.

\subsection{Discrete sampling probabilities and the monotonicity property}

To start with, let's assume that both $s(\cdot)$ and $h(\cdot)$
take discrete values, then $W_{i}$ takes discrete values too since
$I_{i},J_{i},Y_{i}$ are all discrete. The lemma below affirms that
for multi-stage importance sampling with discrete probabilities (e.g. thinking of stratified sampling, where candidates from the same stratum are assigned the same probability), the
rate estimator $\hat{\theta}$ follows a weighted sum of independent
Poisson model. 
\begin{lem}
\label{lem:weighted-poisson}Let $w_{0}^{*}\equiv0<w_{1}^{*}<\cdots<w_{K}^{*}$
enumerate all possible values of $W_{i}$. For $k=0,1,\cdots,K$, let
\begin{align*}
X_{k} & =\sum_{i=1}^{N}I(W_{i}=w_{k}^{*}).
\end{align*}
Then the model \eqref{eq:CP} can be rewritten as below
\begin{align}
\hat{\theta} & \equiv\sum_{k=1}^{K}w_{k}^{*} X_{k},\label{eq:weighted-poisson}
\end{align}
and moreover $X_{k}$, $1\leq k\leq K$, are mutually independent Poissons:
\begin{align*}
X_{k} & \sim Poisson(\lambda_{k}),
\end{align*}
where $\lambda_{k}=\lambda\times P(W_{1}=w_{k}^{*})$. 
\end{lem}

The proof is provided in the appendix. The weighted sum of independent
Poisson model has been studied extensively in the statistical literature, see e.g. \cite{dobson1991confidence,fay1997confidence,fay2017confidence,ng2008confidence,swift2010simulation,tiwari2006efficient}.
Let $q_{\alpha}(Z)$ denote the $\alpha$-quantile of the random
variable $Z$.

Let $(w_{1},\cdots,w_{n})$ be a realization of $\{W_{i}:W_{i}>0\}$
according to \eqref{eq:theta-hat}. Then $x_{k}=\sum_{i=1}^{n}I(w_{i}=w_{k}^{*})$
is the corresponding realization of $X_{k}$ for \eqref{eq:weighted-poisson}
for $k=1,\cdots,K$. Then for a class of Gamma methods, the two-sided
CI for $\theta$ with confidence level $1-\alpha$ can be described by $[q_{\alpha/2}(G_{L}),q_{1-\alpha/2}(G_{U})]$,
where
\begin{itemize}
\item $G_{L}$ follows a Gamma distribution with mean $\sum_{k=1}^{K}w_{k}^{*}x_{k}$
and variance $\sum_{k=1}^{K}w_{k}^{*}{}^{2}x_{k}$, and 
\item $G_{U}$ follows a Gamma distribution with mean $w^{**}+\sum_{k=1}^{K}w_{k}^{*}x_{k}$
and variance $w^{**}{}^{2}+\sum_{k=1}^{K}w_{k}^{*}{}^{2}x_{k}$. 
\end{itemize}
Here $w^{**}$ is a tuning parameter, sometimes referred to as the ``next weight'', for which the original Gamma
CI uses $||w||_{\infty}\equiv w_{K}^{*}$, i.e. the
maximum weight value (\cite{fay1997confidence}), the minimum positive weight (i.e. $w_1^*$) and the average of minimum and maximum weights (i.e. $(w_1^* + w_K^*)/2$) have also been considered in the literature (see \cite{ng2008confidence}), while the modified
Gamma CI (\cite{tiwari2006efficient}) uses $\sqrt{\sum_{k=1}^{K}w_{k}^{*}{}^{2}/K}$ for the variance part and $\sum_{k=1}^{K}w_k^*/K$ for the mean part of $G_U$. The Mid-P Gamma CI (\cite{fay2017confidence}) uses the quantiles of $BG_L + (1-B)G_U$ with $w^{**}$ being the maximum weight, where $B\sim Bernoulli(1/2)$.
Though no theoretical proof is available yet, the original Gamma CI has shown guaranteed coverage (equal or above the nominal level) for all numerical studies
(see e.g. \cite{swift2010simulation,ng2008confidence}).

\begin{table}
\caption{An example where original/modified/mid-p Gamma CIs (lower bound)
all violate the monotonicity property: Category A observes 100 events
all with weight 1 while Category B observes 1 event with weight 100,
both from the same 1M miles. Rounding to the nearest integer, their 90\% two-sided CI lower bounds (ranging from 68 to 81)
for A$\cup$B are smaller than that for A only (84). Note that the weighted Gamma method satisfies monotonicity by returning the lower bound of 103.}\label{tab:counter-example}
\centering{}%
\begin{tabular}{|c|c|c|}
\hline 
 & IPMM for A only  & IPMM for A$\cup$B \tabularnewline
 & (point estimate: 100) & (point estimate: 200) \tabularnewline
\hline 
\hline 
Original Gamma CI \citep{fay1997confidence}  & {[}\textcolor{blue}{84}, 118{]}  & {[}\textcolor{red}{68}, 565{]}\tabularnewline
\hline 
Modified Gamma CI \citep{tiwari2006efficient}  & {[}\textcolor{blue}{84}, 118{]}  & {[}\textcolor{red}{68}, 480{]}\tabularnewline
\hline 
Mid-p Gamma CI \citep{fay2017confidence}  & {[}\textcolor{blue}{84}, 118{]}  & {[}\textcolor{red}{81}, 502{]}\tabularnewline
\hline 
Weighted Gamma CI  & {[}\textcolor{blue}{84}, 118{]}  & {[}\textcolor{blue}{103}, 576{]}\tabularnewline
\hline 
\end{tabular}
\end{table}

Interestingly, in applications, the original Gamma CI can produce counter-intuitive results when we are interested in reporting the rates of two categories (say A and B) of events, as well as their combined rate. For instance, we may want to estimate the total rate of congestion-causing events, but also the rates at which such events are associated with, separately, a delay of less than one minute or a delay of at least one minute. As an illustrative example, suppose we observe 100 events in category A all with weight 1, and 1 event in Category B with weight 100. The point estimates for the rates of events of type A and A$\cup$B will be, respectively 100 and 200.
With the original Gamma method, the lower bound of a 90\% two-sided CI is 84 for Category A; we would expect the lower bound for Category A$\cup$B to therefore be {\em higher} than 84, but instead the original Gamma method yields 68. For the business stakeholder, it would be quite confusing to hear that the rate of 1+ minute congestion events is likely no less than 84 per million miles, but that the rate of {\em all} congestion events might be as low as 68 per million miles.  This intuition is formalized via the following property. 

\begin{defn}
A CI method for rate estimation meets the \emph{monotonicity property}
if its CI bounds increase with the addition of new types of events with weights greater than 0.
\end{defn}

The monotonicity property is intuitive and also important for real
applications due to better interpretability—when two different rates
are summed up, the overall
rate is known to be higher than each individual rate, thus the CI
bounds for the overall rate should be correspondingly higher. In fact, none of the existing Gamma methods have this property, see Table \ref{tab:counter-example} for the counter
example described above. Such violation often happens when the importance sampling quality is uneven, for example, very good at predicting one type of events (i.e. discovered lots of events with small weights), but very hard at predicting another type of events (discovered very few events but with high weights).

\begin{rem}\label{rem:PB-parametric}
Since \eqref{eq:weighted-poisson} is a parametric model where
$X_{k}$ are Poisson, the standard parametric bootstrap procedure
would generate a bootstrap sample as $\hat{\theta}^{(b)}=\sum_{k=1}^{K}w_{k}^{*}X_{k}^{(b)}$,
where $X_{k}^{(b)}\sim Poisson(x_{k})$. Note that $\sum_{k=1}^{K}w_{k}^{*}X_{k}^{(b)}\sim\sum_{i=1}^{n}w_{i}P_{i}^{(b)}$
with $P_{i}^{(b)}\sim Poisson(1)$, i.i.d.. This is equivalent to the Poisson
bootstrap procedure proposed in \cite{bohm2014statistics,chamandy2012estimating} as
mentioned earlier, which produces the confidence interval by the quantiles of ${\hat{\theta}^{(b)}: b=1,\cdots, B}$ for some $B$, say $10^4$. For convenience of reference, Poisson bootstrap is summarized in Table \ref{tab:pb-eb}.
\end{rem}

It is worth pointing out that the Poisson bootstrap CI does satisfy the
monotonicity property. However, it can severely under-estimate the uncertainty
when the observed true positives are rare, which will be demonstrated later in numerical studies in Section 6.

\subsection{Weighted Gamma CI}

Here we propose a weighted Gamma method, which is derived in the spirit
of the fiducial argument as described in \cite{fisher1935fiducial,hannig2016generalized} among others, 
see the detailed derivation in the appendix. The proposed two-sided CI with
confidence level $1-\alpha$ is given by $[L,U]$ with
\begin{align}
L & =q_{\alpha/2}(\sum_{k=1}^{K}w_{k}^{*}G_{k})\label{eq:weighted-gamma-lower-bound}
\end{align}
and
\begin{align}
U & =q_{1-\alpha/2}(w^{**}G_{K+1}+\sum_{k=1}^{K}w_{k}^{*}G_{k})\label{eq:weighted-gamma-upper-bound}
\end{align}
where $G_{1},\cdots,G_{K+1}$ are mutually independent Gamma random
variables, with
\begin{align*}
G_{k} & \sim Gamma(shape=x_{k},rate=1)
\end{align*}
for $k=1,\cdots,K$, and 
\begin{align*}
G_{K+1} & \sim Gamma(shape=1,rate=1).
\end{align*}
Since both $\sum_{k=1}^{K}w_{k}^{*}G_{k}$ and $w^{**}G_{K+1}+\sum_{k=1}^{K}w_{k}^{*}G_{k}$
are weighted sum of independent Gamma random variables, we call this method
``weighted Gamma'', as a new member of the class of Gamma methods (\cite{fay2017confidence}).
For the derivation of \eqref{eq:weighted-gamma-lower-bound} and \eqref{eq:weighted-gamma-upper-bound} in the appendix, we need $w^{**}=||w||_{\infty}=w_{K}^{*}$
(the theoretical maximum weight); however, we also discuss alternative
options later in Section \ref{subsubsection:choice-next-weight}.

\begin{thm}
Under the model \eqref{eq:weighted-poisson}, both Poisson bootstrap and weighted Gamma CI have the monotonicity property as described in Definition 4.
\end{thm}

The proof is obvious and thus omitted.

Though the original Gamma method does not have the monotonicity property while weighted Gamma does, numerical studies suggest that
for most cases weighted Gamma performs very similarly to original Gamma. Part of the reason may be that
the distributions which produce the lower (upper) bounds have the same mean $\sum_k w_k^*x_k$ ($w^{**}+\sum_k w_k^*x_k$) and variance $\sum_k w_k^{*2}x_k$ ($w^{**2} + \sum_k w_k^{*2}x_k$) under the two methods.

\subsection{Exponential bootstrap}\label{subsection:eb}

To extend the weighted Gamma method to continuous sampling probabilities,
one may imagine the set $\{w^*_k: k=1,\cdots,K\}$ converging to a finer and finer lattice with $K=\infty$ in Lemma \ref{lem:weighted-poisson}. As the distribution of $W$ tends towards a continuous distribution, intuitively each
$X_{k}$ becomes either 0 or 1. Since the distribution of $Gamma(shape=1,rate=1)$ is 
the same as  $Exponential(1)$, it is natural to construct a CI as follows:
\begin{itemize}
\item $q_{\alpha/2}(\sum_{i=1}^{n}w_{i}e_{i})$ for the lower bound, and
\item $q_{1-\alpha/2}(\sum_{i=1}^{n}w_{i}e_{i}+w^{**}e_{n+1})$ for the
upper bound,
\end{itemize}
where $e_{i}$, $i=1,\cdots,n+1$, are i.i.d. $Exponential(1)$. There
is no closed-form expression for quantiles of the weighted sum of independent exponential
random numbers, but this can be implemented by the Monte Carlo procedure
described in Table \ref{tab:pb-eb}, which we call ``exponential
bootstrap'' (EB), to mimic Poisson bootstrap. 
We also find that a fast algorithm for EB based on the saddlepoint approximation (see the appendix A.3) 
 performs extremely well,
which is useful for large-scale numerical studies.

\begin{rem}
It is also possible to extend traditional Gamma methods described above to the continuous sampling probability setting. We omit the details here, but present these novel methods as alternatives in the numerical studies of Section \ref{sec:sim}. We found that these CIs perform well in many contexts, but similarly to their predecessors fail to satisfy the monotonicity property.
\end{rem}

\begin{table}
\caption{Poisson bootstrap vs Exponential bootstrap with observed weights $\{w_1, \cdots, w_n\}$ (default $B=10^4$)}
\label{tab:pb-eb}

\centering{}%
\begin{tabular}{|c|>{\centering}p{4.7cm}|>{\centering}p{6.5cm}|}
\hline 
  & Poisson bootstrap (PB) & Exponential bootstrap (EB) with next weight $w^{**}$\tabularnewline
\hline 
\hline 
Procedure  & \begin{raggedright}
For $b=1,\cdots,B$: 
\par\end{raggedright}
\begin{enumerate}[(i)]
\item 
\begin{raggedright}
Draw $e_{i}\sim Poisson(1)$, i.i.d. for $1\leq i\leq n$. 
\par\end{raggedright}
\item 
\begin{raggedright}
Calculate: 
\par\end{raggedright}
\begin{itemize}
\item 
\begin{raggedright}
$\hat{\theta}^{(b)}=\sum_{i=1}^{n}w_{i}e_{i}$ if $n>0$ else
0 
\par\end{raggedright}
\end{itemize}
\end{enumerate}
 & \begin{raggedright}
For $b=1,\cdots,B$: 
\par\end{raggedright}
\begin{enumerate}[(i)]
\item 
\begin{raggedright}
Draw $e_{i}\sim Exponential(1)$, i.i.d. for $1\leq i\leq n+1$. 
\par\end{raggedright}
\item 
\begin{raggedright}
Calculate: 
\par\end{raggedright}
\begin{enumerate}[(a)]
\item 
\begin{raggedright}
$\hat{\theta}_{L}^{(b)}=\sum_{i=1}^{n}w_{i}e_{i}$ if $n>0$ else
0 
\par\end{raggedright}
\item \raggedright{}$\hat{\theta}_{U}^{(b)}=\hat{\theta}_{L}^{(b)}+w^{**}e_{n+1}$ 
\end{enumerate}
\end{enumerate}
\tabularnewline
\hline 
\begin{raggedright}Output \par\end{raggedright} & 
\begin{raggedright}$(1-\alpha)$ CI: \par\end{raggedright} \begin{itemize}
\item Lower bound: $\frac{1}{2}\alpha$ quantile of $\{\hat{\theta}^{(b)}:1\leq b\leq B\}$ 
\item Upper bound: $(1-\frac{1}{2}\alpha)$ quantile of $\{\hat{\theta}^{(b)}:1\leq b\leq B\}$. 
\end{itemize}
 & \begin{raggedright}$(1-\alpha)$ CI: \par\end{raggedright} 
\begin{itemize}
\item Lower bound: $\frac{1}{2}\alpha$ quantile of $\{\hat{\theta}_{L}^{(b)}:1\leq b\leq B\}$ 
\item Upper bound: $(1-\frac{1}{2}\alpha)$ quantile of $\{\hat{\theta}_{U}^{(b)}:1\leq b\leq B\}$. 
\end{itemize}
\tabularnewline
\hline 
\end{tabular}
\end{table}

\subsubsection{Choice of next weight $w^{**}$}\label{subsubsection:choice-next-weight}

For importance sampling,  the theoretical maximum weight $||w||_{\infty}$
is often too large even if defensive techniques (e.g. \cite{hesterberg1995weighted, owen2000safe}) are properly used, which would make the EB upper bound too conservative to be useful. Choosing a value of $w^{**}$ which provides the right coverage without being too conservative is difficult in general, as these properties depend on both $\theta$ and nuisance parameters (e.g. those governing the distribution of $W_1$). Motivated by the original Gamma and modified
Gamma methods, we recommend
\begin{align}\label{eq:w**}
w_m & =\max(\{||w||_{2},w_{1},\cdots,w_{n}\})
\end{align}
where $||w||_{2}$ is the square root of theoretical second moment defined by
\begin{align}\label{eq:w2}
||w||_{2} & =\sqrt{E(W_{1}^{2}|W_{1}>0)}.
\end{align}
In other words, $w_m$ is the maximum value of observed maximum weight and $||w||_2$. 

Here $||w||_{2}$ depends on the distribution of unknown true positive probabilities
$r(V_{1})$ and thus needs to be estimated. We recommend an estimation
method below by assuming a relationship between $r(v)$ and $p(v)$:
\begin{align}\label{eq:power-law}
p(v) & \propto r(v)^{\gamma}
\end{align}
where $\gamma$ is the index and $\gamma=1/2$ corresponds to optimal importance sampling in
terms of variance minimization (e.g. \cite{liu2019estimating}), while $\gamma$ close to 0 corresponds to near uniform sampling. 
Equivalently,
this suggests a model
\begin{align*}
r(v) & \propto p(v)^{1/\gamma}
\end{align*}
which may be used to fit the index $\gamma$, yielding an estimate $\hat{\gamma}$, if there is enough properly chosen historical data--otherwise one may set $\hat{\gamma}=1/2$, which per numerical studies often works reasonably well when sampling is not too far from optimal.

Let $\hat{r}(v)\propto p(v)^{1/\hat{\gamma}}$. Recall that  $p(v_{i})$
is only available if segment $i$ is simulated (i.e. $I_{i}=1$), see Remark \ref{rem:h(X)}. Since $E(W_1^2|W_1>0) = E(r(V_1)/p(V_1))/E(r(V_1)p(V_1))$, we may construct
a Hajek style estimator as below:
\begin{align} \label{eq:estimate-2nd-moment}
\hat{E}(W_{1}^{2}|W_{1}>0) & =\frac{\sum_{i:I_{i}=1}s(v_{i})^{-1} \hat{r}(v_i)/p(v_i)}{\sum_{i:I_{i}=1}s(v_{i})^{-1}\hat{r}(v_i)p(v_{i})}
\end{align}
which can be calculated since the unknown scaling factor in $\hat{r}$ cancels out between the numerator and denominator. Of note is that the estimator \eqref{eq:estimate-2nd-moment} makes use of sampling probabilities (and estimated weights) for all run segments that were sampled for simulation, not just those which proceeded to the human review stage.

\section{Numerical Studies}

\label{sec:sim}

In this section we report various statistical simulation studies to evaluate the
CI methods, with a focus on the PB, EB, and extensions of some traditional Gamma methods, for two particular choices of next weight discussed above.
The setup tries to mimic the data collection procedure
described in Section 2 with notations introduced in Section 3. For
simplicity and without loss of generality, we set $s(v)\equiv1$ for all $v$, i.e. all run segments
are simulated, so that multi-stage importance sampling reduces to
single-stage importance sampling. In the appendix A.4, 
we report the results of a two-stage sampling simulation to validate that the findings presented here also generalize to that setting.

Let $f_{1}$ be the density function for the features from true positive
candidate events and $f_{0}$ for the false positive candidates. Let $\pi$
be the proportion of true positives among all candidates. Let $\lambda$
be the expected  number of candidates per million
miles. Recall $r(v)=P(Y=1|v)$, i.e. the probability that a candidate associated
with feature $v$ is a true positive, then by Bayes' theorem $r(v)=\pi f_{1}(v)/(\pi f_{1}(v)+(1-\pi)f_{0}(v))$.
Let $b\in(0,1)$ be the budget ratio for sampling (i.e. smaller $b$
means less human review) and let $p(v)$ be the sampling function.
A statistical simulation scenario is determined
by $(\lambda,\pi,f_{1},f_{0}, p, b)$ with true rate value $\theta=\pi\lambda$. In practice, the connection between $p$ and $r$ is unknown; Below, we first report studies where $p$ is parameterized by $r$ and an index parameter $\gamma$ with varying values according to \eqref{eq:power-law}.

For each scenario, the statistical simulation is performed as below:
\begin{enumerate}[(1)]
\item Generate $N\sim Poisson(\lambda)$ as the total number of run segments with features $(V_{1},\cdots,V_{N})$
and labels $(Y_{1},\cdots,Y_{N})$, where $(Y_{i},V_{i})$ are i.i.d.
according to $Y_{i}\sim Bernoulli(\pi)$ (1 for true positive and
0 for false positive) and $V_{i}\sim f_{1}$ if $Y_{i}=1$, otherwise
$V_{i}\sim f_{0}$.
\item Let $B_{i}\sim Bernoulli(p(V_{i}))$ indicate whether candidate-$i$
is sampled or not, for $i=1,\cdots,N$, where 
$p(V_{i})=\min(1,Nb\times\overline{p}_{i})$
with $\overline{p}_{i} \propto r(V_{i})^{\gamma}$ such that $\sum_i\overline{p}_i=1$. So, roughly $p(v) \propto r(v)^{\gamma}$.
\item Let $W_{i}=\frac{B_{i}}{p(V_{i})}\times Y_{i}$ if $B_{i}=1$ and 0 otherwise.
\item Report the point estimate $\hat{\theta}=\sum_{i=1}^{N}W_{i}$ and
CI for each method.
\end{enumerate}
For each scenario, the procedure is replicated $10^4$ times with different random seeds, and the empirical coverage error (i.e. the fraction of replicates whose CIs fail to cover the true value of $\theta$)
and average CI width for each CI method are reported.

Some type of events are relatively common, and the observed event
counts will be relatively high, while other event types are more rare.
Here, we report results with $f_{1}$ and $f_{0}$
set to the normal density functions for $\mathcal{N}(2,2^{2})$ and $\mathcal{N}(-2,2^{2})$ respectively,
with $\lambda=10^{6}$ and $\pi=10^{-3}$, which implies $\theta=1000$.
We vary the budget ratio parameter $b$ to cover both common and rare
cases. We report results where $b$ ranges from 0 to 0.05, which corresponds
to the expected number of true positives from 0 to 50 under uniform
sampling. 

We use the index parameter $\gamma$ to control the sampling quality, which
is set to either 0.1, 0.5, or 0.9. Note that $\gamma=0.5$ is considered optimal
sampling, for instance in \cite{liu2019estimating}. In the range $\gamma < 0.5$, which approaches simple random sampling at $\gamma=0$, we describe sampling as {\em defensive}. In such a regime, the risk of missing interesting classes of events because of ML model errors is reduced. By contrast, the regime $\gamma > 0.5$ is characterized as {\em greedy} sampling, since it samples more heavily than optimal for candidates with high true positive probabilities and down-samples others.

For EB, we consider two choices of next weight: estimated $||w||_{2}$
and estimated $w_{m}$ as described in Section \ref{subsubsection:choice-next-weight},
for which $\hat{\gamma}$ is set to $\gamma$ (i.e. we assume that there is enough historical data to fit it). The corresponding two EB CIs are
denoted as EB2 and EB2m respectively in the plots below. In addition to comparing PB with the EB methods, we also study the choice of $\hat{\gamma}$ in estimating $w^{**}$, and report results which suggest $\hat{\gamma}=0.5$ performs well
for most cases but can perform poorly for some range of settings.

Similar to the extension from weighted Gamma to EB, other traditional Gamma methods can also be extended to the continuous case by considering infinitely many strata. In this section we  include the comparison with the following two extensions with the same notation as described in Section 5.1: (1) $[q_{\alpha/2}(G_L), q_{1-\alpha/2}(G_U)]$, labeled as GO (extending the \textbf{\textit{O}}riginal \textbf{\textit{G}}amma CI), and (2) using the $\alpha/2$ and $1-\alpha/2$ quantiles of $BG_L + (1-B)G_U$, labeled as GP (extending the Mid-\textbf{\textit{P}} \textbf{\textit{G}}amma CI), both with suffix 2m indicating the next weight $w^{**}$ being the same as that for EB2m; Using the same next weight as for EB2 results in somewhat less reliable coverage and thus is omitted below. 

Note that the monotonicity property is not evaluated here, but as shown before, both EB and PB guarantee the property, while similar to traditional Gamma methods, neither GO nor GP does. GO and GP are included as references regarding the coverage and CI width comparison.

\subsection{Simulation for uniform sampling}

\label{subsec:sim_opt_uniform}
We start with the baseline example $\gamma=0$, which is equivalent
to uniform sampling, in which case all weights are the same and thus
GO2m, EB2 and EB2m are the same, essentially equivalent to the "exact" Poisson CI \citep{garwood1936fiducial}. Figure \ref{fig:v3_scale0} reports the
empirical coverage errors and average CI widths for PB, EB, and GP. The
coverage of EB is more reliable than PB overall while the coverage error for PB can be very large and even close to 100\% when true positive data
is sparse, and both converge to the nominal level when the sampling
rate goes up. EB's CI width is larger than PB when true positive
data is sparse, and the width gets closer between the two methods
when true positive data gets denser. GP is somewhat between PB and EB.

\begin{figure}
\includegraphics[width=0.8\textwidth]{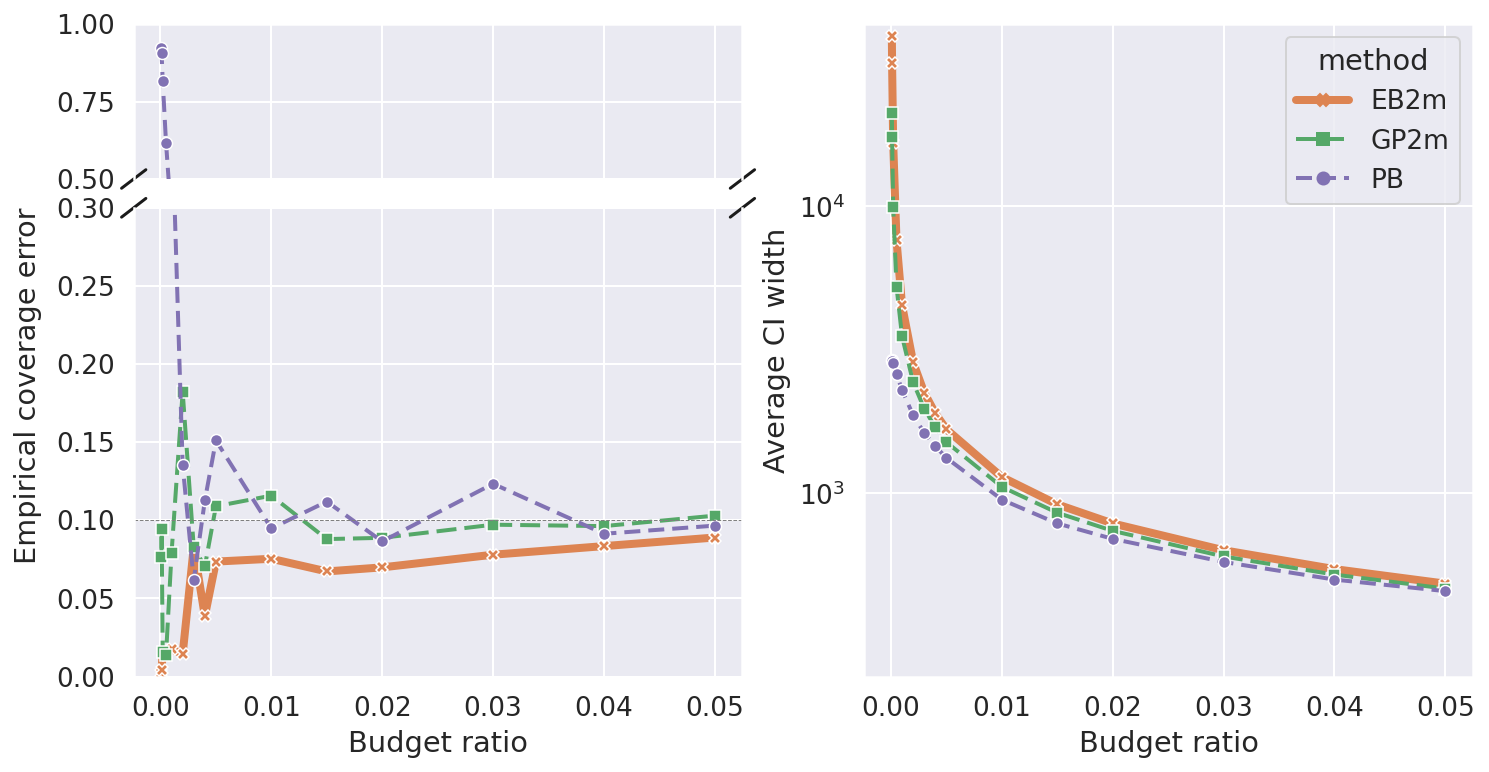} \centering
\caption{ Comparison between PB, GP2m and EB2m in terms of empirical coverage error
(left panel) and average CI width (in log scale, right panel) of $90\%$ two-sided
confidence intervals, where $\gamma=0$, i.e. uniform
sampling, and the budget ratio ranges from 0 to 5\%. Note that GO2m, EB2 and EB2m are identical here.}\label{fig:v3_scale0}
\end{figure}

\subsection{Simulation for importance sampling}

\label{subsec:sim_opt_importance}

Next, we simulate actual importance sampling by varying the value of $\gamma$ to be 0.1, 0.5, or 0.9. For the simulation in this section we set $\hat\gamma$ to be equal to $\gamma$, which is reasonable if there is enough historical data to estimate $\hat\gamma$.

\begin{figure}[hbt!]
  \centering
   \begin{subfigure}[b]{0.8\textwidth}
        \centering
        \caption{$\gamma=0.1$}
    \includegraphics[width=\textwidth]{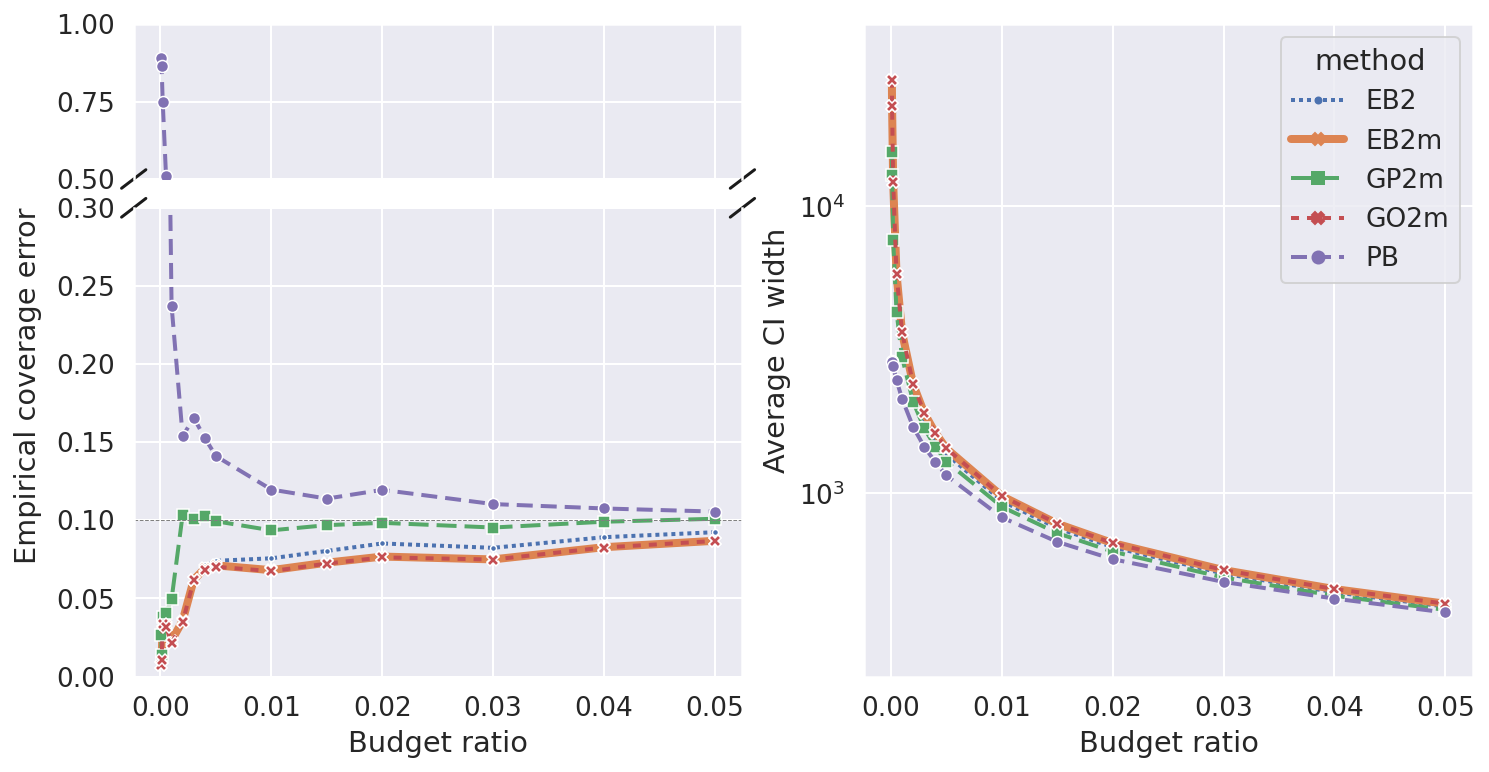}
   \end{subfigure}
   \begin{subfigure}[b]{0.8\textwidth}
         \centering
         \caption{$\gamma=0.5$}
         \includegraphics[width=\textwidth]{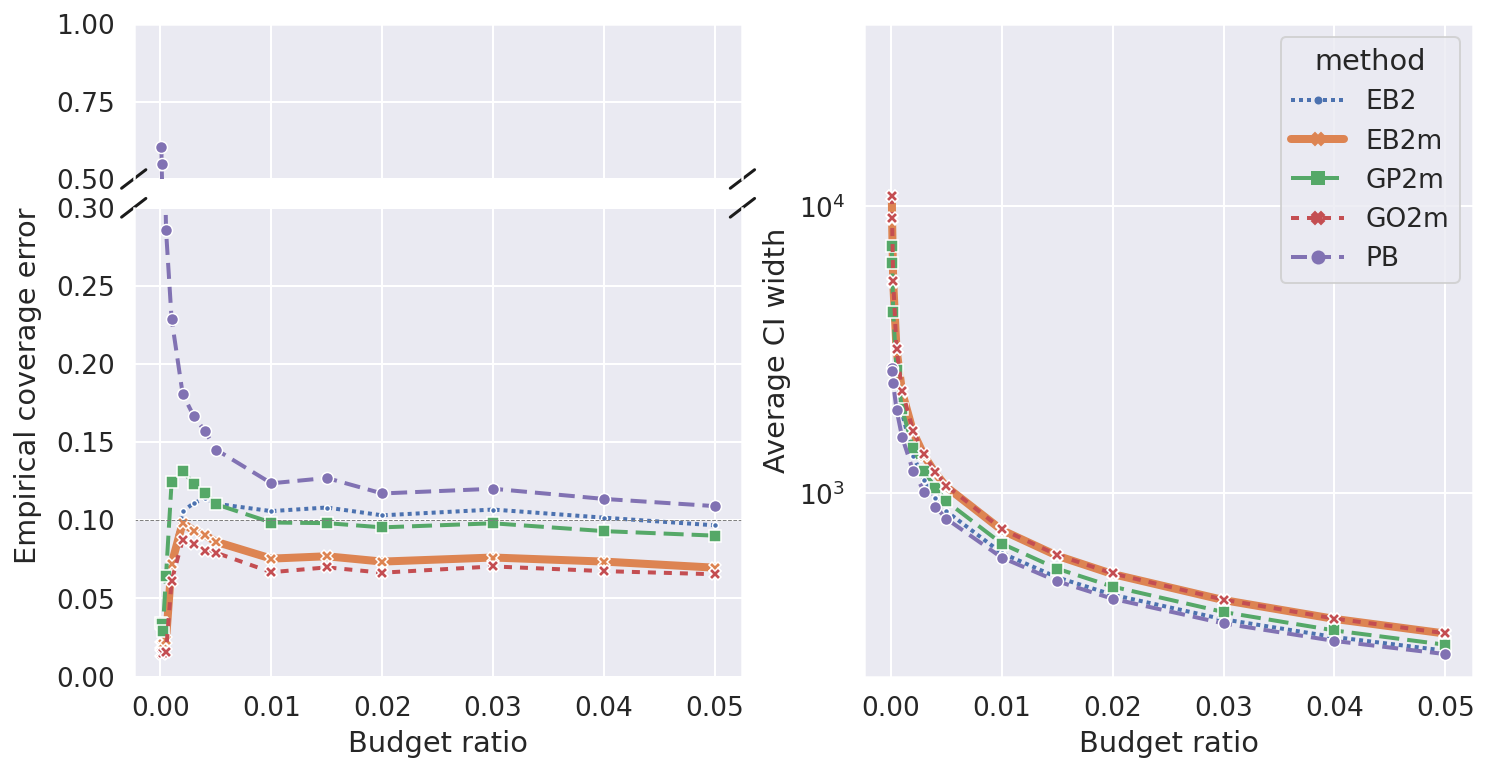}
   \end{subfigure}
   \begin{subfigure}[b]{0.8\textwidth}
         \centering
         \caption{$\gamma=0.9$}
         \includegraphics[width=\textwidth]{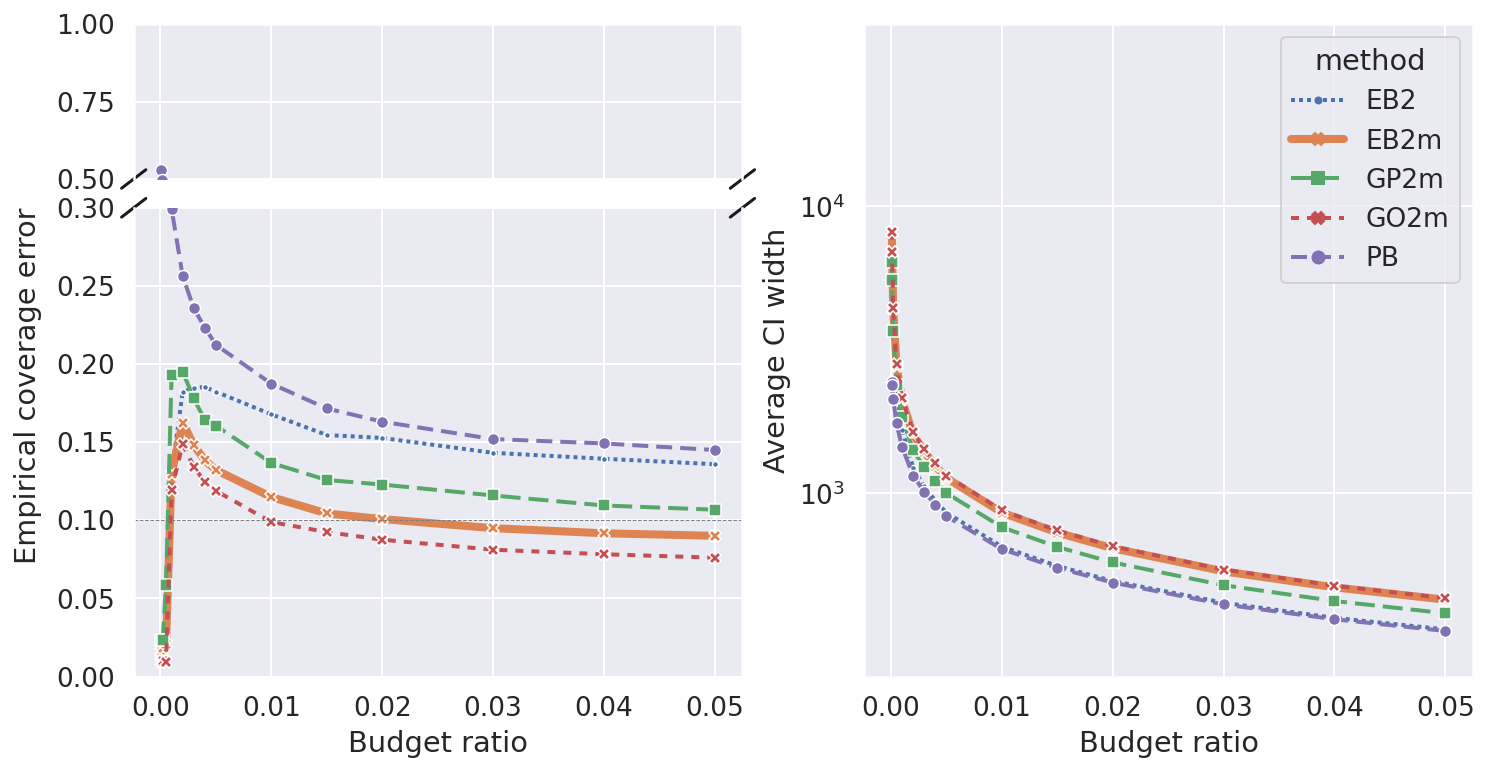}
     \end{subfigure}
  \centering 
  \caption{Comparison between PB, GO2m, GP2m, EB2 and EB2m w.r.t. empirical coverage error (left panels) and average CI width (right panels, in log scale) of $90\%$ two-sided CIs with varying $\gamma$ values, where the budget ratio ranges from 0 to 5\%.}\label{fig:sparse}
\end{figure}

Figure \ref{fig:sparse} reports the performance comparison between PB, GO2m, GP2m, EB2 and EB2m. 
The results show that the coverage error for PB can
go up to almost 100\% as the sampling budget gets close to 0, while EB2m maintains
the nominal coverage for most of the cases. For $\gamma=0.9$, EB2m over-covers for budget close to 0, then somewhat
under-covers but quickly recovers as budget ratio
increases. Nevertheless in sparse cases, EB2m is a clear winner over PB: it maintains
coverage at the cost of a larger CI width, but the width gap
drops fast as budget ratio increases. We note that EB2 performs similarly to EB2m but often has less reliable coverage. We also included the performance of GO and GP with the "winning next weight" (i.e. the one for EB2m instead of EB2) -- GP2m performs the best when $\gamma=0.1$ and for most cases when $\gamma=0.5$ but when $\gamma=0.9$, it significantly undercovers and performs worse than both EB2m and GO2m. Overall, EB2m and GO2m are quite comparable.

In applications where the monotonicity property is not important (for example where there is no desire to estimate rates for disjoint categories of events), the extended mid-point Gamma method (GP2m) provides a good alternative (and even enjoy superior coverage properties) to the Exponential Bootstrap, especially when sampling is defensive.

\subsection{Choice of $\hat{\gamma}$}

In reality, there may not always be enough
data to fit $\hat{\gamma}$; in such cases users may use domain knowledge to set $\hat{\gamma}$, or use 0.5 by assuming the sampling quality to be near optimal. 
However, users should be cautious
with the choice of $\hat{\gamma}=0.5$ without dedicated investigation of the sampling quality. Figure \ref{fig:compare}
compares the performance of EB2m with $\hat{\gamma}=0.5$ and
$\hat{\gamma}=\gamma$, for different values of $\gamma$. The results suggest that $\hat{\gamma}=0.5$ performs
well when $\gamma=0.1$ (closer to uniform sampling), but can perform poorly with empirical coverage errors even above 35\% (while the nominal level is 10\%)
for some small budget ratio in the range from 0.06\% to 0.2\% when $\gamma=0.9$ (much
more greedy than optimal sampling).

\begin{figure}[hbt!]
     \centering
     \begin{subfigure}[b]{0.8\textwidth}
         \centering
         \caption{$\gamma=0.1$}
         \includegraphics[width=\textwidth]{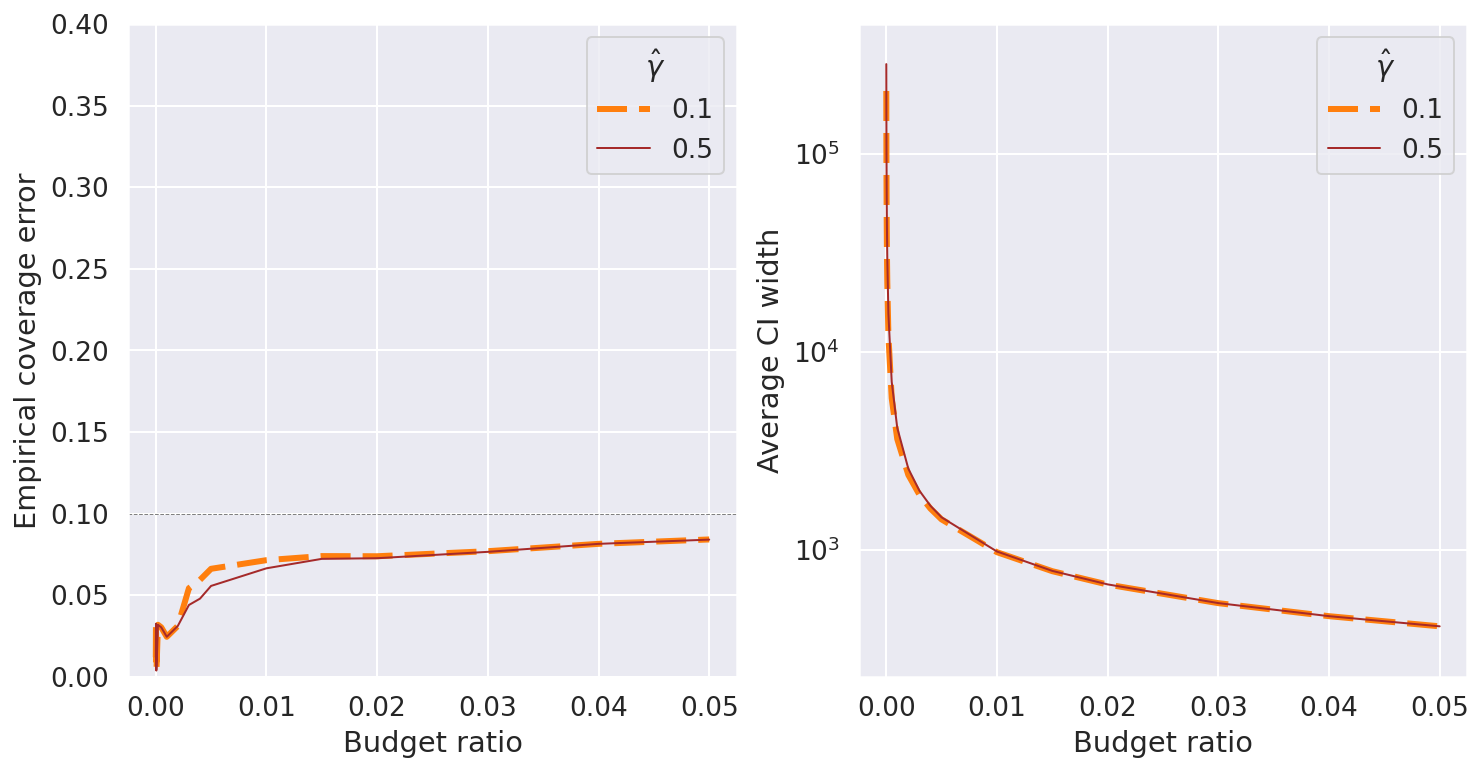}
         
     \end{subfigure}
          \begin{subfigure}[b]{0.8\textwidth}
         \centering
         \caption{$\gamma=0.9$}
         \includegraphics[width=\textwidth]{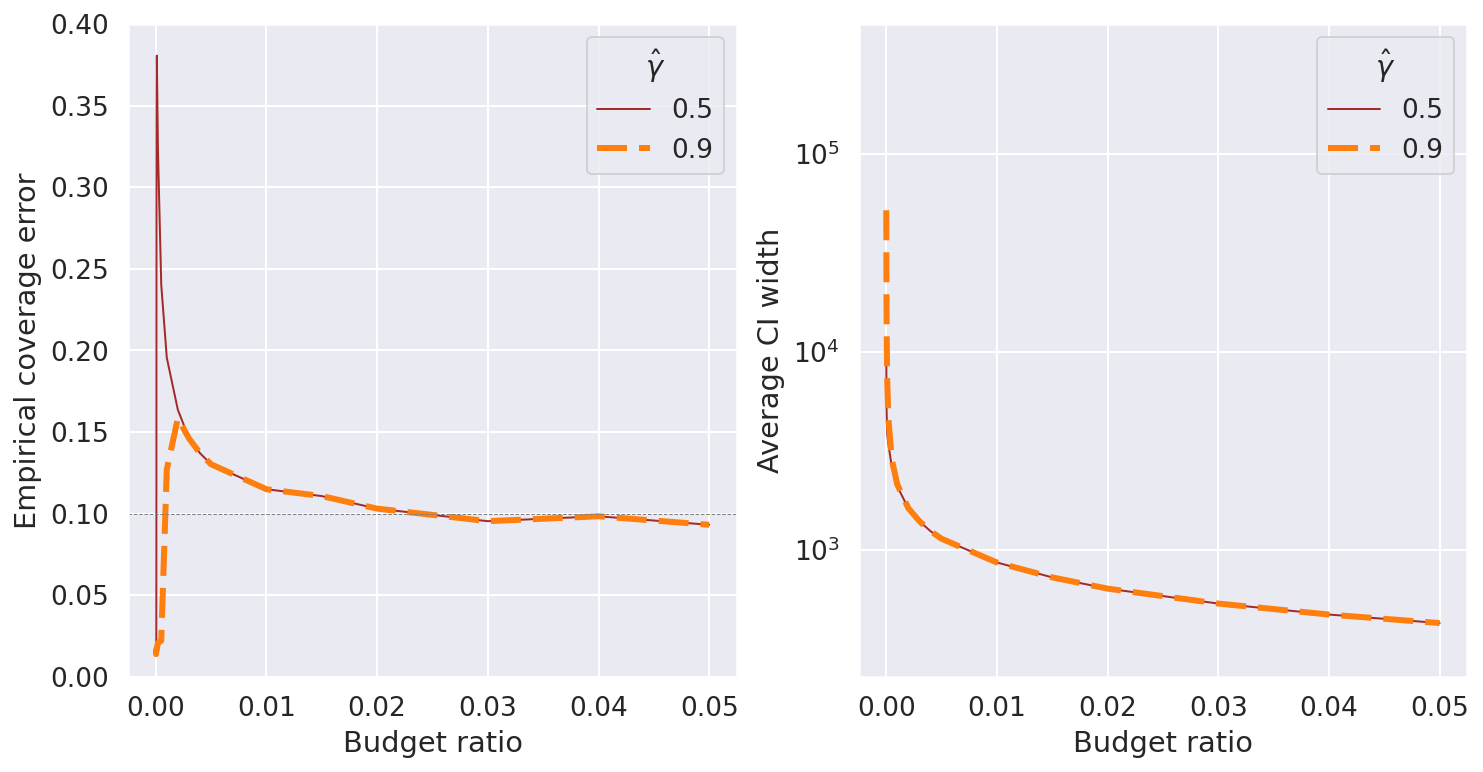}
         
     \end{subfigure}
\caption{Comparison of EB2m between $\hat{\gamma}=0.5$ and $\hat{\gamma}=\gamma$ in terms of empirical coverage error (left panels) and average CI width (right panels, in log scale) of $90\%$ two-sided CIs with varying $\gamma$ values, where the budget ratio ranges from 0 to 5\%.} \label{fig:compare}
     \end{figure}

\subsection{Simulation for "misspecified" models} The preceding studies show that EB2m works well when sampling follows the power relationship, i.e. \eqref{eq:power-law}, and the power index $\gamma$ is not much greater than the optimal value of 0.5. While this relationship may be reasonable (as practitioners can adaptively design and adjust sampling models to approximate it in end-to-end sampling to estimation applications), further investigation indicates that EB2m's performance is quite robust to misspecifications of this relationship. Nevertheless, it can perform poorly when the sampling is considerably more greedy than optimal sampling, unless the budget ratio is large enough; This finding aligns with earlier observations.

Figure \ref{fig:mis} reports the performance comparison under two "misspecified" scenarios:
\begin{enumerate}[(a)]
\item $p(v)  \propto \sqrt{r(v)}(1-r(v))$: In this case, sampling is near optimal for small $r(v)$ but adversarial for $r(v)$ close to 1. 
\item $p(v)  \propto r(v)(1+r(v))$: Here sampling is considerably more greedy than optimal. 
\end{enumerate}
To construct the CIs, $||w||_2$ is estimated according to \eqref{eq:estimate-2nd-moment} assuming $\hat{\gamma}=0.5$. 
The results show that EB2m and GO2m exhibit similar performance under both (a) and (b) while EB2 and PB are nearly identical under (b); PB systematically undercovers under both (a) and (b), on the other hand, the coverage of EB2 and GP2m is very close to the nominal level under (a) for most of the range of budget ratios, but they perform systematically worse than both EB2m and GO2m under (b) for the majority of the range; Furthermore, under model (a) EB2m and GO2m still maintain reasonable coverage even if the budget ratio is small, whereas under model (b), none of the methods do unless the budget ratio is large enough.

\begin{figure}[hbt!]
     \centering
     \begin{subfigure}[b]{0.8\textwidth}
         \centering
         \caption{$p(v)  \propto \sqrt{r(v)}(1-r(v))$}
         \includegraphics[width=\textwidth]{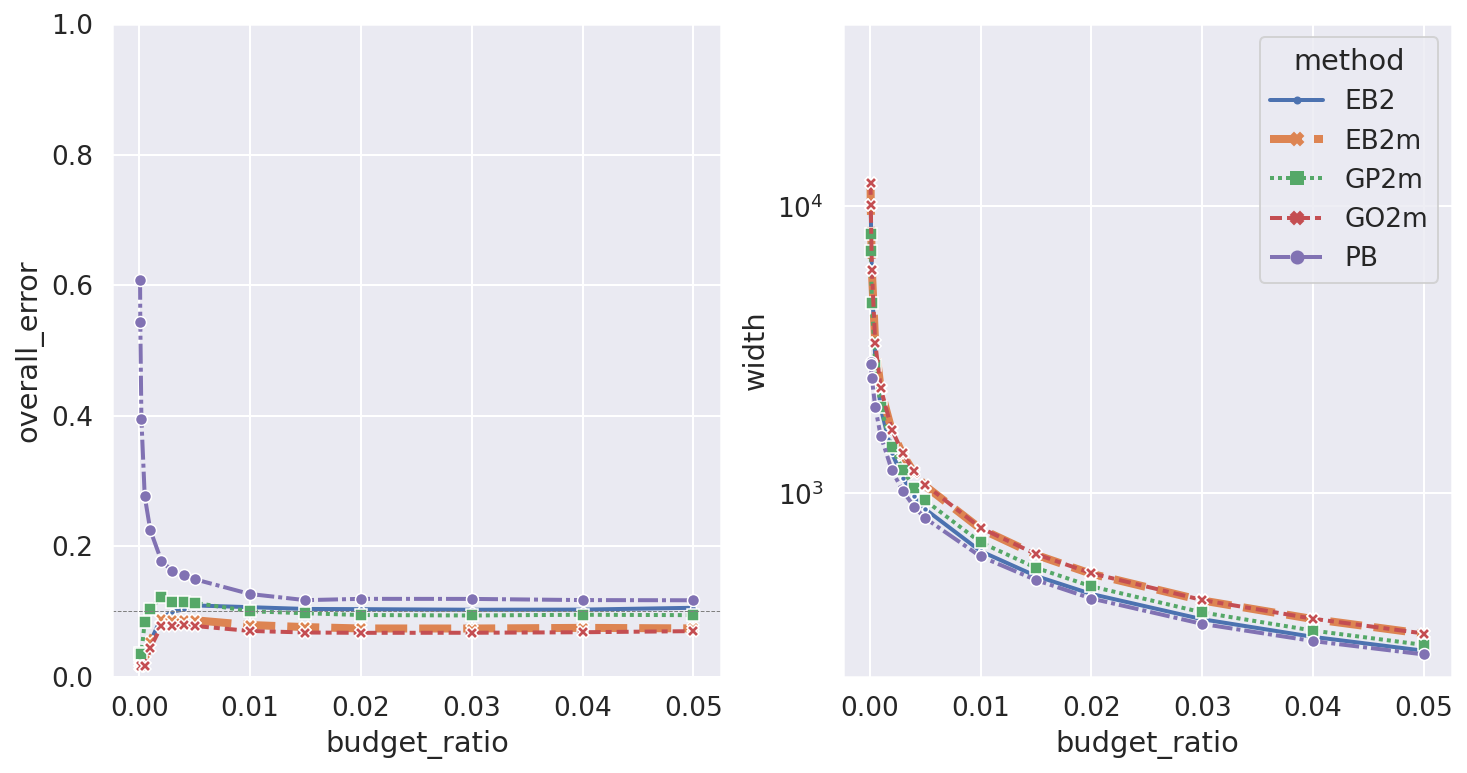} 
         
     \end{subfigure}
          \begin{subfigure}[b]{0.8\textwidth}
         \centering
         \caption{$p(v)  \propto r(v)(1+r(v))$}
         \includegraphics[width=\textwidth]{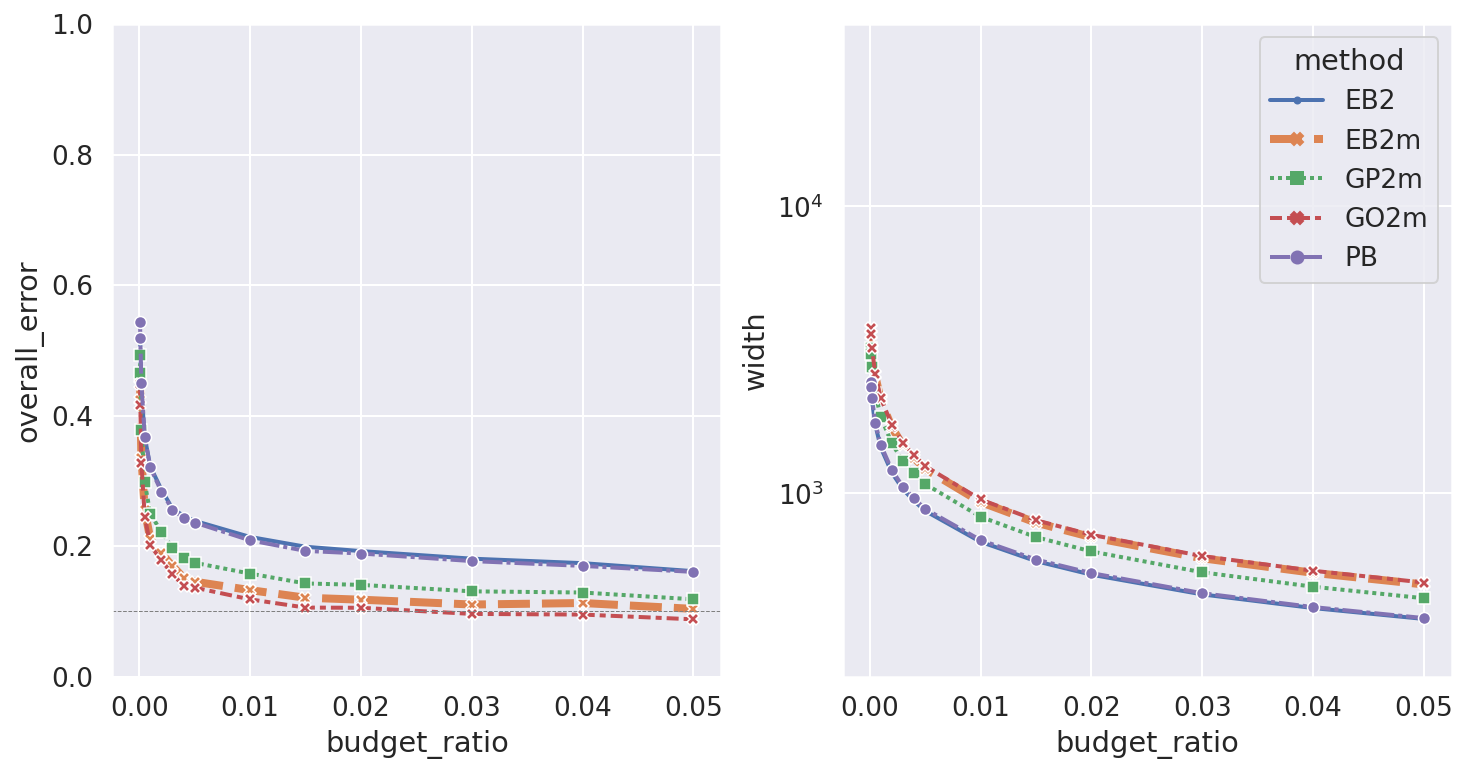} 
         
     \end{subfigure}
\caption{Performance comparison w.r.t. empirical coverage error (left panels) and average CI width (right panels, in log scale) of $90\%$ two-sided CIs, under different relationships between the true event rates $r(v)$ and the sampling probabilities $p(v)$, where EB2 and GP2m are close to each other under (a), EB2 and PB nearly overlap under (b), and EB2m and GO2m perform similarly under both (a) and (b).} \label{fig:mis}
     \end{figure}

\label{subsec:mis}

\section{Real-world data analysis}

In this section, we report some
results from a real case study which consists of millions of
simulations and tens of thousands of human reviews with the budget ratio less than 1\%. We omit a detailed description of the application and actual metric for confidentiality reasons, but the problem structure and data collection follow the framework described in Sections 2-3 above.
For simulations which
produced candidate events and thus were eligible for human review,
the distribution of overall sampling probabilities (i.e. $p(V_i)$) is very heavy-tailed, as reported in Figure \ref{fig:Histograms-of-sampling},
where the vertical axis shows the counts in the log scale (exact y-axis values are hidden). The jumpy portion between 0.8 and 0.9 is likely due to
covariate shift between machine learning training data and this specific
test data, a common phenomenon in machine learning (see, e.g. \cite{sugiyama2007covariate}).

\begin{figure}
\caption{Histogram of overall sampling probabilities for simulation units
which produced candidate events.}\label{fig:Histograms-of-sampling}

\includegraphics[width=3in,totalheight=2in]{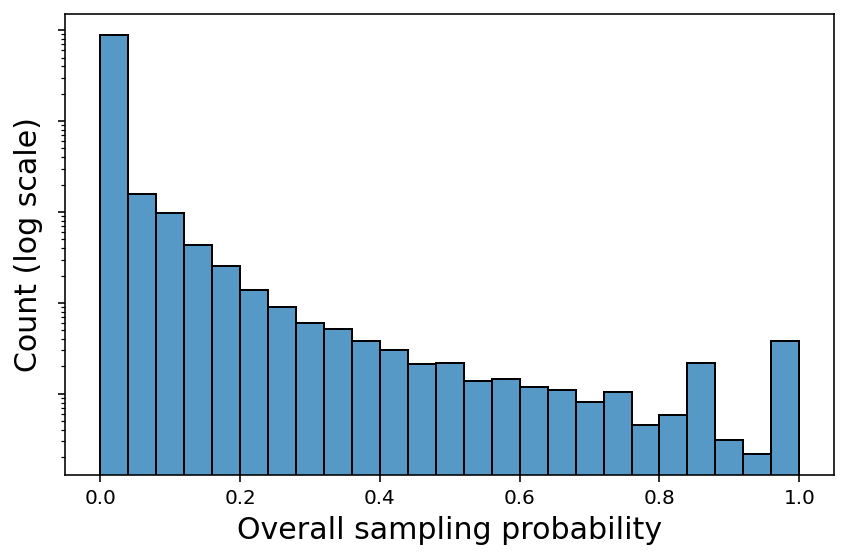}
\end{figure}

In this case study, there were two disjoint groups of events which we denote simply by Category A and Category B. This partitioning of events led to a violation of the monotonicity property by extensions of the traditional Gamma CIs, similar to the
toy example reported in Table 1. In this case, 38 true
positive events were identified from Category A, all having
relatively small sampling weights; on the other hand, a single event with a high weight was identified from
Category B. The exact weights as well as the next weight estimate for this data set are provided
in the appendix. The comparison of CIs between the rate for Category A and
the rate for A$\cup$B is reported for EB2m, PB,
GO2m and GP2m methods in Table \ref{tab:CI-comparison-real-world}.
Note that the rates are rescaled (i.e. without normalization by the actual
mileage) for the purpose of anonymization. 
The results show that for A events only, EB2m is quite close to GO2m for both lower bounds and upper bounds, 
while for A$\cup$B, their upper bounds are quite close but GO2m and GP2m 
have much smaller lower bounds than EB2m. On the other hand,
the upper bounds of PB are significantly smaller than the other three methods. 
It is also evident that
the monotonicity property is violated by the extension of traditional Gamma CIs; for example,
at the 90\% confidence level, GO2m violates this property (its lower bound for A exceeds that for A$\cup$B), while at the
95\% confidence level, both GO2m and GP2m violate this property. 
As expected,
the monotonicity property is always preserved by EB2m. PB also satisfies monotonicity, but our numerical studies suggest that it likely suffers from undercoverage due to heavy-tailed sampling probabilities (and rare events). 

\begin{table}
\caption{CI comparison with real-world examples, where 90\%, 95\% and 99\%
CIs are reported for GO2m, GP2m, Poisson bootstrap (PB) and Exponential
bootstrap CI (EB2m) and the CI bounds are rounded to nearest integers}\label{tab:CI-comparison-real-world}

\begin{tabular}{|c|c|c|}
\hline 
 & rescaled IPMM for A only   & rescaled IPMM for A$\cup$B \tabularnewline
 & (point estimate: 231) & (point estimate: 385) \tabularnewline
\hline 
\hline 
GO2m (90\%)  & {[}\textcolor{blue}{148}, 468{]}  & {[}\textcolor{red}{141}, 2035{]}\tabularnewline
\hline 
GP2m (90\%)  & {[}155, 426{]}  & {[}185, 1792{]}\tabularnewline
\hline 
EB2m (90\%)  & {[}149, 472{]}  & {[}226, 2060{]}\tabularnewline
\hline 
PB (90\%)  & {[}149, 323{]}  & {[}171, 1372{]}\tabularnewline
\hline 

\hline 
GO2m (95\%)  & {[}\textcolor{blue}{135}, 507{]}  & {[}\textcolor{red}{103}, 2322{]}\tabularnewline
\hline 
GP2m (95\%)  & {[}\textcolor{blue}{141}, 468{]}  & {[}\textcolor{red}{134}, 2077{]}\tabularnewline
\hline 
EB2m (95\%)  & {[}137, 524{]}  & {[}202, 2379{]}\tabularnewline
\hline 
PB (95\%)  & {[}134, 344{]}  & {[}157, 1445{]}\tabularnewline
\hline 

\hline 
GO2m (99\%)  & {[}\textcolor{blue}{113}, 590{]}  & {[}\textcolor{red}{51}, 2952{]}\tabularnewline
\hline 
GP2m (99\%)  & {[}\textcolor{blue}{115}, 556{]}  & {[}\textcolor{red}{67}, 2706{]}\tabularnewline
\hline 
EB2m (99\%)  & {[}116, 642{]}  & {[}165, 3094{]}\tabularnewline
\hline 
PB (99\%)  & {[}109, 383{]}  & {[}137, 1814{]}\tabularnewline
\hline 

\end{tabular}
\end{table}

\section{Discussion}

In this paper we introduced the problem of rate estimation with multi-stage importance sampling, an important one in the AV industry, with estimating rates of AV-caused congestion events as a motivating example. We derived a unified semi-parametric compound Poisson formulation, which may be applied more broadly than AV, especially in the era of artificial intelligence (AI). Furthermore, we proposed a novel exponential bootstrap CI method that features a desirable monotonicity property and devised a data-driven choice of the “next weight” parameter. We presented numerical studies showing that EB performs well for a wide range of settings which are relevant to our applications, while some incumbent methods fail to satisfy either coverage or monotonicity requirements; When the monotonicity requirement is not needed, the extended mid-point Gamma method (GP2m) provides a good alternative to EB2m, especially when sampling is defensive.

Our numerical studies suggest that the choice of $\hat\gamma$ or more generally the next weight $w^{**}$ may deserve further study, for example, when the sampling may be sub-optimally designed. This highlights an open problem: how to quantify or diagnose sampling quality. The simple $\gamma$-model we proposed, along with its parameter estimation, offers an initial step towards addressing this issue. For instance, a $\gamma$ value significantly greater than 0.5 corresponds to ``greedy'' sampling, and more greedy sampling (indicated by larger $\gamma$) would increase the discovery of events of interest but at the cost of larger or even infinite variance (equivalent to $||w||_2=\infty$) implying a wide CI in rate estimation. 
In cases of extreme misspecification of the model (5.6) or significant covariate shift, the variance of the Horvitz-Thompson estimator can explode, or $\gamma$ can be negative, leading to uninformative wide CIs.
Another important applied problem is the construction of CIs for the difference between two rates, for which it may be interesting to note that the compound Poisson model can be extended with the weights taking both positive and negative values. 

\section*{Significance Statement}

Evaluating autonomous vehicle (AV) performance requires accurately estimating the rate of very rare events. This challenging problem relies on large-scale data, sampled via complex algorithms from hundreds of millions of miles of real and simulated driving logs. This paper formalizes the estimation problem, develops a new statistical method for rigorous inference, and introduces a novel monotonicity criterion for uncertainty quantification which is logically expected from a causal perspective and provides enhanced interpretability in applied settings. AV evaluation represents a new paradigm: leveraging massive real and synthetic data with heavy sampling to understand long-tail phenomena. We foresee these "needle-in-a-haystack" problems as an emerging theme in applied statistics, and this work provides foundational contributions to the domain.


\section*{Acknowledgements}

We would like to thank Vitya Aleksandrov, Chuanwen Chen, Yin-Hsiu Chen, Yiqun Chen, Kevin Donoghue, Joyce Guo, Lin He, Luna Huang, Peter Lau,
Claire McLeod, Jingang Miao, Ben Sherman, Yuanbiao Wang, Kelvin Wu, Azeem Zaman, and Xiaoyue Zhao for many insightful discussions from problem formulation, technical development to operation, and Aman Sinha for suggesting an elegant binary search technique. We further appreciate Editor Lexin Li, Associate Editor, and anonymous reviewers whose comments have helped improve the paper significantly.

%
%
 


\bibliographystyle{imsart-nameyear} 
\bibliography{mpci_ci.bib}       


\begin{appendix}

\section{}

\subsection{Proof of Lemma 3} 
\begin{proof}
The equality of (5.1) 
follows from basic algebra.
Next, the probability generating function of $X_{k}$ can be written
as
\begin{align*}
E(t^{X_{k}}) & =E(E(t^{X_{k}}|N))=E((E(t^{I(W_{1}=w_{k}^{*})}))^{N})=E((f_{k}t+(1-f_{k}))^{N}).
\end{align*}
Since $N\sim Poisson(\lambda)$, then 
\begin{align*}
E(t^{X_{k}}) & =e^{\lambda(f_{k}t+(1-f_{k})-1)}=e^{\lambda f_{k}(t-1)}
\end{align*}
which coincides with the generating function of $Poisson(\lambda f_{k})$.
Thus $X_{k}\sim Poisson(\lambda f_{k})$.

Finally, to prove that $\{X_{k}:k=0,\cdots,K\}$ are mutually independent,
let's look at the multivariate generating function: 
\begin{align*}
E(\prod_{k=0}^{K}t_{k}^{X_{k}}) & =E((E(\prod_{k=0}^{K}t_{k}^{I(W_{1}=w_{k}^{*})})^{N})=E((\sum_{k=0}^{K}f_{k}t_{k})^{N})=e^{\lambda(\sum_{k=0}^{K}f_{k}t_{k}-1)}.
\end{align*}
Due to $\sum_{k=0}^{K}f_{k}=1$, $E(\prod_{k=0}^{K}t_{k}^{X_{k}})=\prod_{k=0}^{K}e^{\lambda f_{k}(t_{k}-1)}=\prod_{k=0}^{K}E(t_{k}^{X_{k}})$,
which completes the proof since $w_{0}^{*}X_{0}\equiv0$.
\end{proof}

\subsection{Derivation of the weighted Gamma method based on the fiducial argument}

For the reader's convenience, the model (5.1) 
is
the weighted sum of independent Poisson random numbers:
\begin{align*}
\hat{\theta} & =\sum_{k=1}^{K}w_{k}^{*}X_{k}
\end{align*}
where $0<w_{1}^{*}<\cdots<w_{K}^{*}$ are known constants, and $X_{k}\sim Poisson(\lambda_{k})$
with unknown $\lambda_{k}$. We would like to construct a confidence
interval for $\theta=E(\hat{\theta}) = \sum_{k=1}^K w_k^* \lambda_k$.

Our derivation of the lower bound and upper bound follows the spirit
of Fisher's fiducial argument \cite{fisher1935fiducial} which states
that ``\emph{In general, it appears that if statistics $T_{1},T_{2},T_{3},\cdots$
contain jointly the whole of the information available respecting
parameters $\theta_{1},\theta_{2},\theta_{3},\cdots$ and if functions
$t_{1},t_{2},t_{3},\cdots$ of the $T$'s and $\theta$'s can be found,
the simultaneous distribution of which is independent of $\theta_{1},\theta_{2},\theta_{3},\cdots$,
then the fiducial distribution of $\theta_{1},\theta_{2},\theta_{3},\cdots$
 simultaneously may be found by substitution.''} Here we find that the “substitution” idea can be extended to inequalities.

\subsubsection{Some basics for the derivation}

Consider a (homogeneous) Poisson process with rate 1, and let $T(n)$ be the arrival
time of the $n$th event for $n=1,2,\cdots$. With some abuse of notation,
let $X$ be the number of events which fall into $(0,\lambda]$, with $\lambda>0$,
i.e. $X=\sum_{n=1}^{\infty}I(T(n)\leq\lambda)$, then $X\sim Poisson(\lambda)$.
The lemma below (see \cite{feller2008introduction})  summarizes some
statistical properties to be used later.

\begin{lemma}
\label{lem:poisson-process}For any non-negative integer $x$, 

\begin{enumerate}[(i)]
\item $X=x$ if and only if $T(x)\leq\lambda<T(x+1)$. 
\item The marginal distribution of $T(x)$ is Gamma with shape $x$ and
scale 1. 
\item For any $0<\lambda_{1}<\lambda_{2}$, $\sum_{n=1}^{\infty}I(T(n)\leq\lambda_{1})$
and $\sum_{n=1}^{\infty}I(\lambda_{1}<T(n)\leq\lambda_{2})$ are independent
Poisson random variables with mean $\lambda_{1}$ and $\lambda_{2}-\lambda_{1}$. 
\end{enumerate}

\end{lemma}

\subsubsection{Derivation of the lower confidence limit}

Now back to the model (5.1) 
with $K$ independent
Poisson random variables, we may consider $K$ independent Poisson
processes, and let $\{T_{i}(n):n\geq1\}$ denote the arrival times
of the $i$th Poisson process with rate 1 associated with $X_{i}$, for $i=1,\cdots,K$.

With observations $X_{i}=x_{i}$, by Lemma~\ref{lem:poisson-process}
we have
\begin{align*}
T_{i}(x_{i}) & \leq\lambda_{i}<T_{i}(x_{i}+1).\label{eq:bound-poisson-rate}
\end{align*}
Therefore, 
\begin{align}
\sum_{i=1}^{K}w_{i}^{*}T_{i}(x_{i}) & \leq\sum_{i=1}^{K}w_{i}^{*}\lambda_{i} < \sum_{i=1}^{K}w_{i}^{*}T_{i}(x_i+1)
\end{align}
which suggests that the lower confidence limit of $\sum_{i}w_{i}^{*}\lambda_{i}$
can be constructed by the $\alpha/2$ quantile of $\sum_{i=1}^{K}w_{i}^{*}T_{i}(x_{i})$. By substitution, the distribution of each $T_i(x_{i})$ is an independent Gamma with shape $x_{i}$
and rate 1. This completes the derivation of (5.2).

Note that the second inequality in \eqref{eq:bound-poisson-rate}  suggests an upper bound
by the $1-\frac{\alpha}{2}$ quantile of $\sum_{i=1}^{K}w_{i}^{*}T_{i}(x_{i}+1)$,
which is however too conservative. One may use the fiducial argument with
different choices of statistics to develop different bounds (see e.g.
\cite{stein1959example}), and indeed we have found a much tighter
upper bound in the next subsection.

\subsubsection{Derivation of the upper confidence limit}

To derive a tighter upper bound, instead of considering $K$ Poisson
processes as above, we need to consider a single Poisson process with rate 1. Again,
let $\{T(n):n=1,2,\cdots\}$ denote the arrival times of the events.

Let $\Lambda_{0}=0$ and $\Lambda_{i}=\sum_{j=1}^{i}\lambda_{j}$
for $i=1,\cdots,K$. Then $\Lambda_{1}\leq\cdots\leq\Lambda_{K}$
since $\lambda_{i}\geq0$.

With some abuse of notation, let $X_{i}$ be the number of events
which fall into $(\Lambda_{i-1},\Lambda_{i}]$: 
\begin{align*}
X_{i} & =\sum_{n=1}^{\infty}I(\Lambda_{i-1}<T(n)\leq\Lambda_{i}).
\end{align*}
Since $\lambda_{i}\equiv\Lambda_{i}-\Lambda_{i-1}$, by Lemma \ref{lem:poisson-process},
$\{X_{i},1\leq i\leq K\}$ are mutually independent Poisson with 
\begin{align*}
X_{i} & \sim Poisson(\lambda_{i}),
\end{align*}
which agrees with our original model (5.1).

Let $S_{i}=\sum_{j=1}^{i}X_{j}$ be the number of events which fall into $(0,\Lambda_{i}]$, then we may rewrite $S_{i}=\sum_{n=1}^{\infty}I(T(n)\leq\Lambda_{i})$.

Given the observation $X_{i}=x_{i}$ for $1\leq i\leq K$, let $s_{i}=\sum_{j=1}^{i}x_{j}$,
and $s_{0}\equiv0$.

By Lemma \ref{lem:poisson-process}, $S_{i}=s_{i}$ 
if and only if
\begin{align*}
T(s_{i}) & \leq\Lambda_{i}<T(s_{i}+1).
\end{align*}

Note that $w_{i+1}^{*}-w_{i}^{*}>0$. With some basic algebra, we
have
\begin{align*}
\sum_{i}w_{i}^{*}\lambda_{i} & =w_{K}^{*}\Lambda_{K}-\sum_{i=1}^{K-1}(w_{i+1}^{*}-w_{i}^{*})\Lambda_{i}\\
 & \leq w_{K}^{*}T(s_{K}+1)-\sum_{i=1}^{K-1}(w_{i+1}^{*}-w_{i}^{*})T(s_{i})\\
 & =w_{K}^{*}(T(s_{K}+1)-T(s_{K-1}))+\sum_{i=1}^{K-1}w_{i}^{*}(T(s_{i})-T(s_{i-1}))\\
 & =\sum_{i=1}^{K}w_{i}^{*}(T(s_{i})-T(s_{i-1}))+w_{K}^{*}(T(s_{K}+1)-T(s_{K})).
\end{align*}
Let $G_{i}=T(s_{i})-T(s_{i-1})$ for $i=1,\cdots,K$ and $G_{K+1}=T(s_{K}+1)-T(s_{K}).$
If we plug in $s_i$s as fixed values in the spirit of the fiducial argument, then by Lemma  \ref{lem:poisson-process}, 
$\{T(s_{i})-T(s_{i-1}):1\leq i\leq K\}$
and $T(s_{K}+1)-T(s_{K})$ are mutually independent Gamma, i.e.
\begin{align*}
G_{K+1} & \sim Gamma(shape=1,rate=1),
\end{align*}
and for $i=1,\cdots,K$, due to $s_{i}-s_{i-1}\equiv x_{i}$, 
\begin{align*}
G_{i} & \sim Gamma(shape=x_{i},rate=1).
\end{align*}

The above inequality suggests by the Fiducial argument to set the
upper bound for $\sum_{i}w_{i}^{*}\lambda_{i}$ by the $(1-\alpha/2)$
quantile of $\sum_{i=1}^{K}w_{i}^{*}G_{i}+w_{K}^{*}G_{K+1}$, which
completes the derivation of (5.3).

\subsection{Fast algorithm by saddlepoint approximation}
\label{subsec:fast-eb}

Since EB can be treated as a special realization of the weighted Gamma
CI where each weight appears at least once, we consider the fast approximation
for the weighted Gamma CI where the tuning parameter $w^{**}$ for
next weight is user-specified. Let $(x_{1},\cdots,x_{K})$ be the
observed values associated with $(w_{1}^{*},\cdots,w_{K}^{*})$. For
notation simplicity, we may use $w_{K+1}^{*}\equiv w^{**}$ and
$x_{K+1}=1$. Then the computation of the weighted Gamma CI bounds
requires the computation of the quantiles of two weighted Gamma random
variables in the form of
\begin{align*}
Z & =\sum_{}w_{i}^{*}G_{i}
\end{align*}
where $G_{i}$s are independent random variables with $G_{i}\sim Gamma(shape=x_{i},rate=1)$.

Let $\kappa(t)=\log E(e^{tZ})$. Let $\phi$ and $\Phi$ be the standard
normal density function and cumulative distribution function, respectively.
Let $sign(t)$ be the sign function which is equal to 1 if $t>0$,
-1 if $t<0$, and 0 if $t=0$. The saddlepoint approximation of the
tail distribution for $Z$ \citep{daniels1954saddlepoint,lugannani1980saddle}
can be described as below:
\begin{itemize}
\item If $z=EZ$,
\begin{align*}
P(Z\geq z) & \approx\frac{1}{2}-\frac{\kappa'''(0)}{6\sqrt{2\pi}\sigma^{3}}
\end{align*}
where $\sigma^{2}=\kappa''(0)=var(Z)$;
\item If $z\neq EZ$,
\begin{align*}
P(Z\geq z) & \approx1-\Phi(\xi)+\phi(\xi)(\omega^{-1}-\xi^{-1})
\end{align*}
where 
\begin{align*}
\kappa'(t^{*}) & =z\\
\omega & =t^{*}\sqrt{\kappa''(t^{*})}\\
\xi & =sign(t^{*})\sqrt{2(t^{*}z-\kappa(t^{*}))}.
\end{align*}
\end{itemize}
Let
\begin{align*}
\omega(t) & =t\sqrt{\kappa''(t)}\\
\xi(t) & =sign(t)\sqrt{2(t\kappa'(t)-\kappa(t))}\\
f(t) & =1-\Phi(\xi(t))+\phi(\xi(t))(\frac{1}{\omega(t)}-\frac{1}{\xi(t)}).
\end{align*}
To find the quantile $z$ such that $P(Z\geq z)=\alpha$, the algorithm
works as below:
\begin{enumerate}[(i)]
\item Solve $t^{*}\in(-\infty,\frac{1}{\max_{i}(w_{i}^{*})})$ such that
$f(t^{*})=\alpha$;
\item Set $z=\kappa'(t^{*})$.
\end{enumerate}
For our specific case, after some basic algebra, we have
\begin{align*}
\kappa(t) & =\sum_{i}-x_{i}\log(1-w_{i}^{*}t)
\end{align*}
\begin{align*}
\kappa'(t) & =\sum_{i}\frac{w_{i}^{*}x_{i}}{1-w_{i}^{*}t}
\end{align*}
\begin{align*}
\kappa''(t) & =\sum_{i}\frac{w_{i}^{*2}x_{i}}{(1-w_{i}^{*}t)^{2}}
\end{align*}
\begin{align*}
\kappa'''(t) & =\sum_{i}\frac{2w_{i}^{*3}x_{i}}{(1-w_{i}^{*}t)^{3}}.
\end{align*}
where $t<\frac{1}{max_{i}(w_{i}^{*})}$. Then $\kappa'(0)=\sum_{i}w_{i}^{*}x_{i}$,
$\sigma^{2}=\kappa''(0)=\sum_{i}w_{i}^{*2}x_{i}$ and $\kappa'''(0)=2\sum_{i}w_{i}^{*3}x_{i}$.

A simple binary search can solve $f(t^{*})=\alpha$ quickly and reliably.
To identify the left and right bounds $[t_{L},t_{R}]$ for bisection,
Dr. Aman Sinha suggests an efficient method as below: Initialize $t=-1$,
iterate $t\leftarrow2\times t$ until $f(t)<\alpha$, then set $t_{L}=t$; Finally
set $t_{R}=t_{L}/2$ if $t_{L}<-1$ otherwise $t_{R}=1/\max_{i}(w_{i}^{*})$.

Our numerical studies (see Figure \ref{fig:Saddlepoint-approximation-example}
as an illustration example) suggest that the saddlepoint approximation
works extremely well, most likely due to some nice property of weighted
Gamma, which is consistent with numerical discoveries in the literature,
e.g. \cite{hesterberg1994saddlepoint} showed numerically that the
saddlepoint approximation works extremely well for standard Gamma.

We also verified in Figure \ref{fig:eb_vs_saddle1} and \ref{fig:eb_vs_saddle2}  that for the range of parameters used in numerical studies, EB2m (with default $B=10^4$) and its saddlepoint approximation perform very similarly.

\begin{figure}

\caption{Illustration of the saddlepoint approximation accuracy  for the weighted Gamma
distribution $G_{100}+10\times G_{10}+100\times G_1$, where $G_{100}, G_{10}, G_1$ are independent Gamma random variables with rate 1 and shape 100, 10, 1 respectively, 
the curve is the cumulative distribution function (CDF) based on Monte-Carlo
with sample size $10^{6}$, and the cross-shaped markers are based on the saddlepoint approximation.}\label{fig:Saddlepoint-approximation-example}

\includegraphics[width=4in]{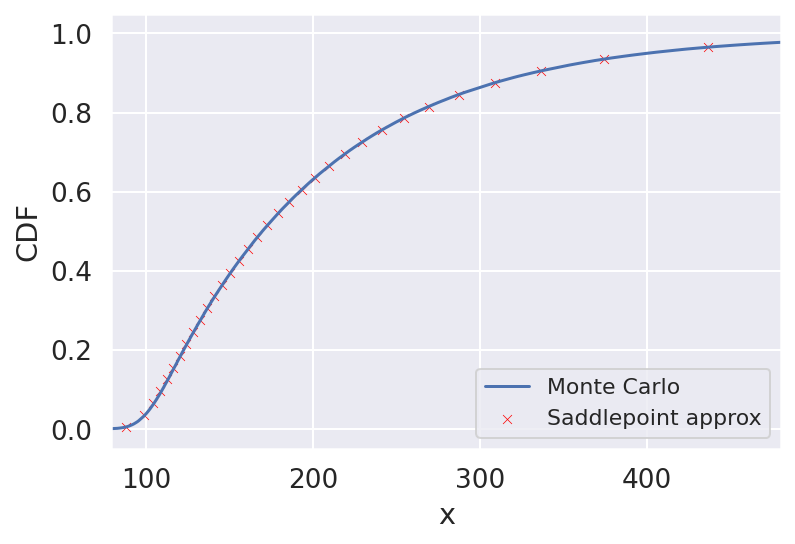}

\end{figure}

\begin{figure}[hbt!]
  \centering
   \begin{subfigure}[b]{0.8\textwidth}
        \centering
        \caption{$\gamma=0.1$}
    \includegraphics[width=\textwidth]{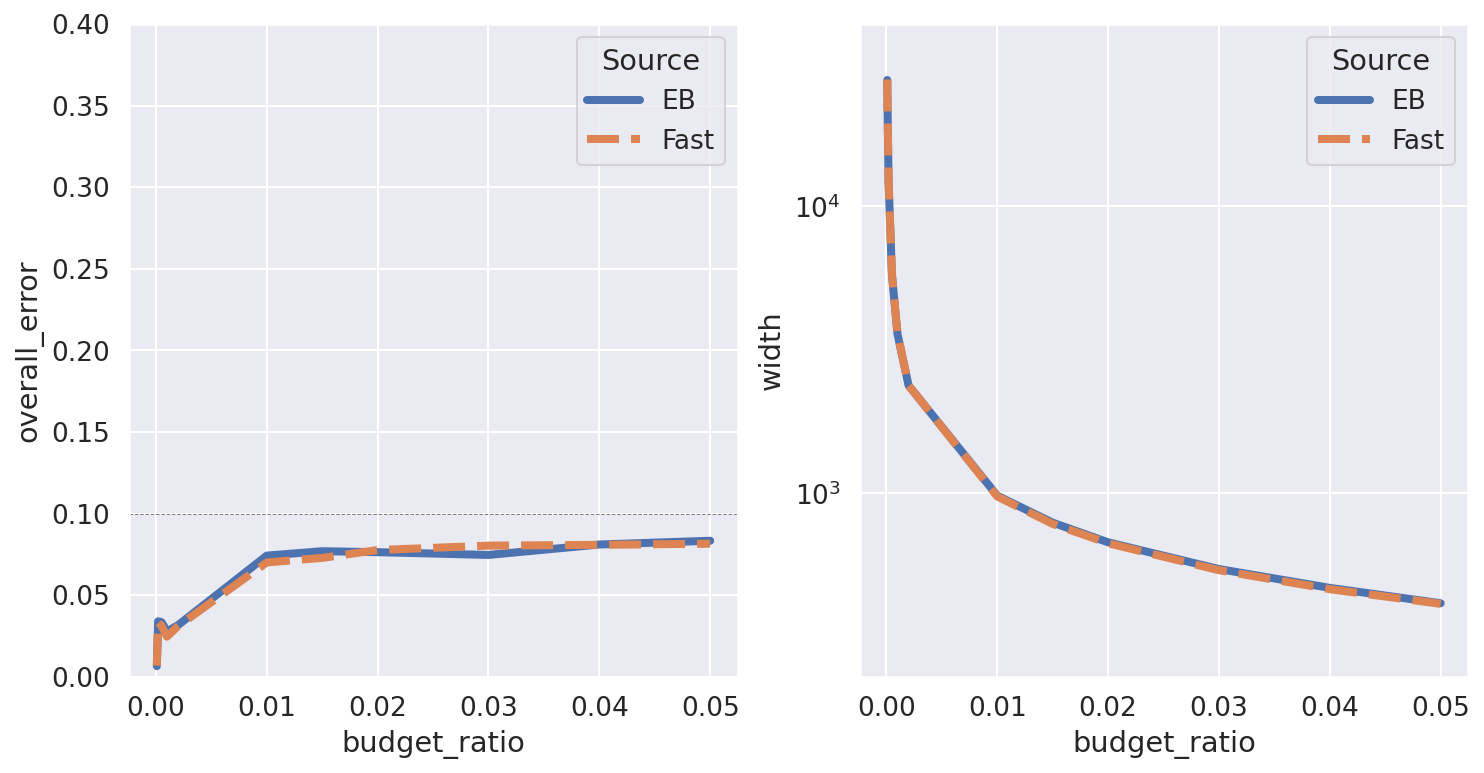}
   \end{subfigure}
   \begin{subfigure}[b]{0.8\textwidth}
         \centering
         \caption{$\gamma=0.5$}
         \includegraphics[width=\textwidth]{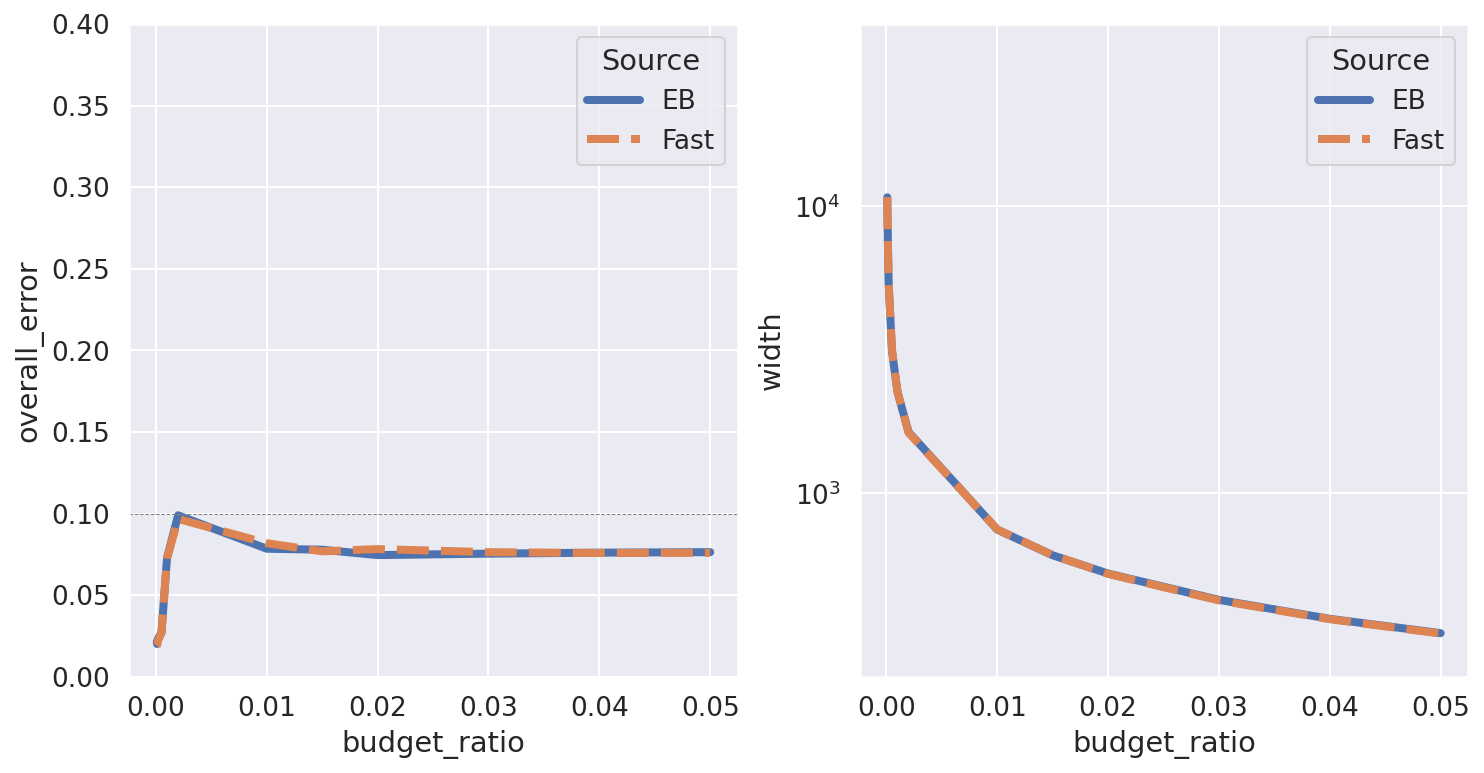}
   \end{subfigure}
   \begin{subfigure}[b]{0.8\textwidth}
         \centering
         \caption{$\gamma=0.9$}
         \includegraphics[width=\textwidth]{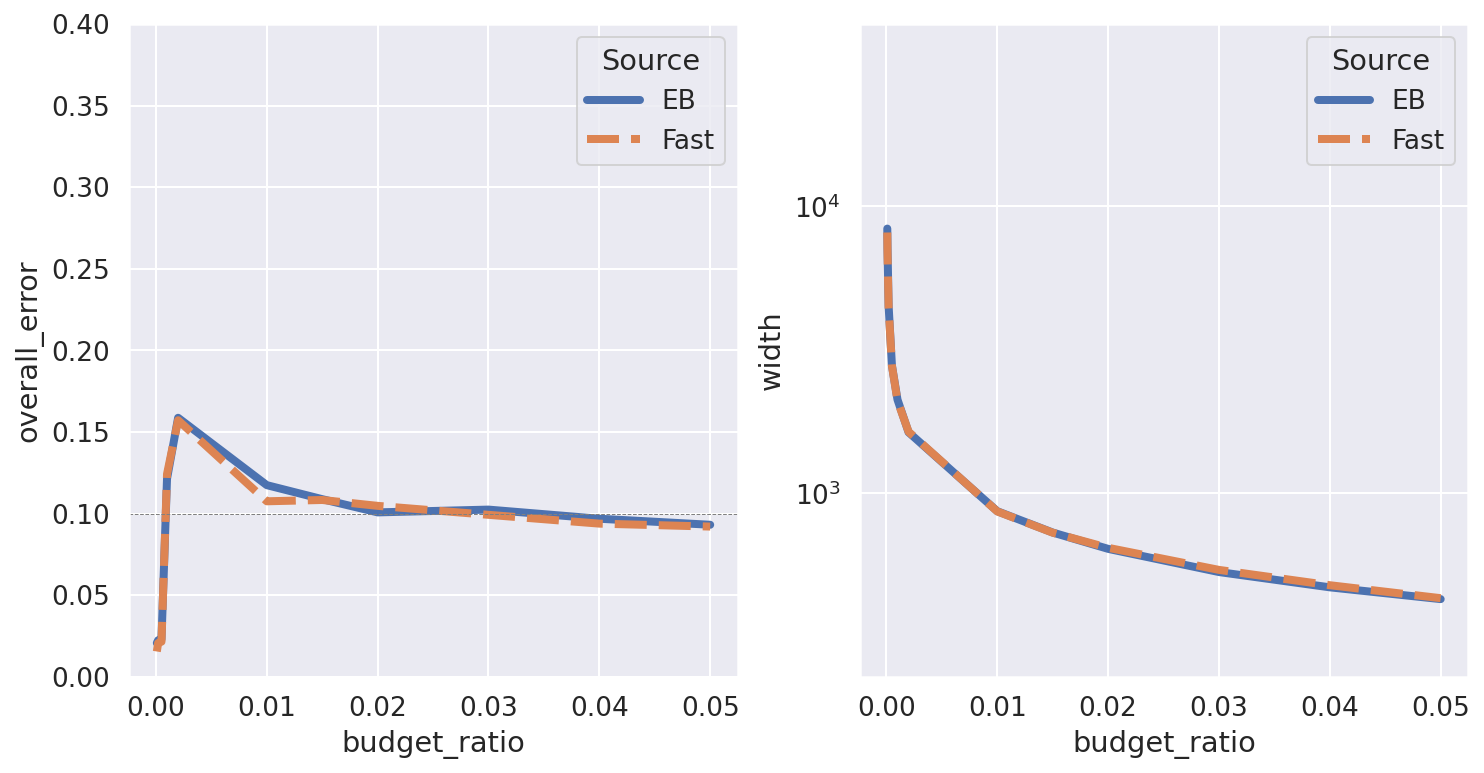}
     \end{subfigure}
  \centering 
  \caption{Comparison between EB2m (labeled as EB) and its saddlepoint approximation (labeled as Fast) w.r.t. empirical coverage error (left panels) and average CI width (right panels, in log scale) of $90\%$ two-sided CIs with varying $\gamma$ values and $\hat{\gamma}=\gamma$, where the budget ratio ranges from 0 to 5\%, same setting as Figure 6.2 
  in the paper.}\label{fig:eb_vs_saddle1}
\end{figure}

\begin{figure}[hbt!]
     \centering
     \begin{subfigure}[b]{0.8\textwidth}
         \centering
         \caption{$\gamma=0.1$}
         \includegraphics[width=\textwidth]{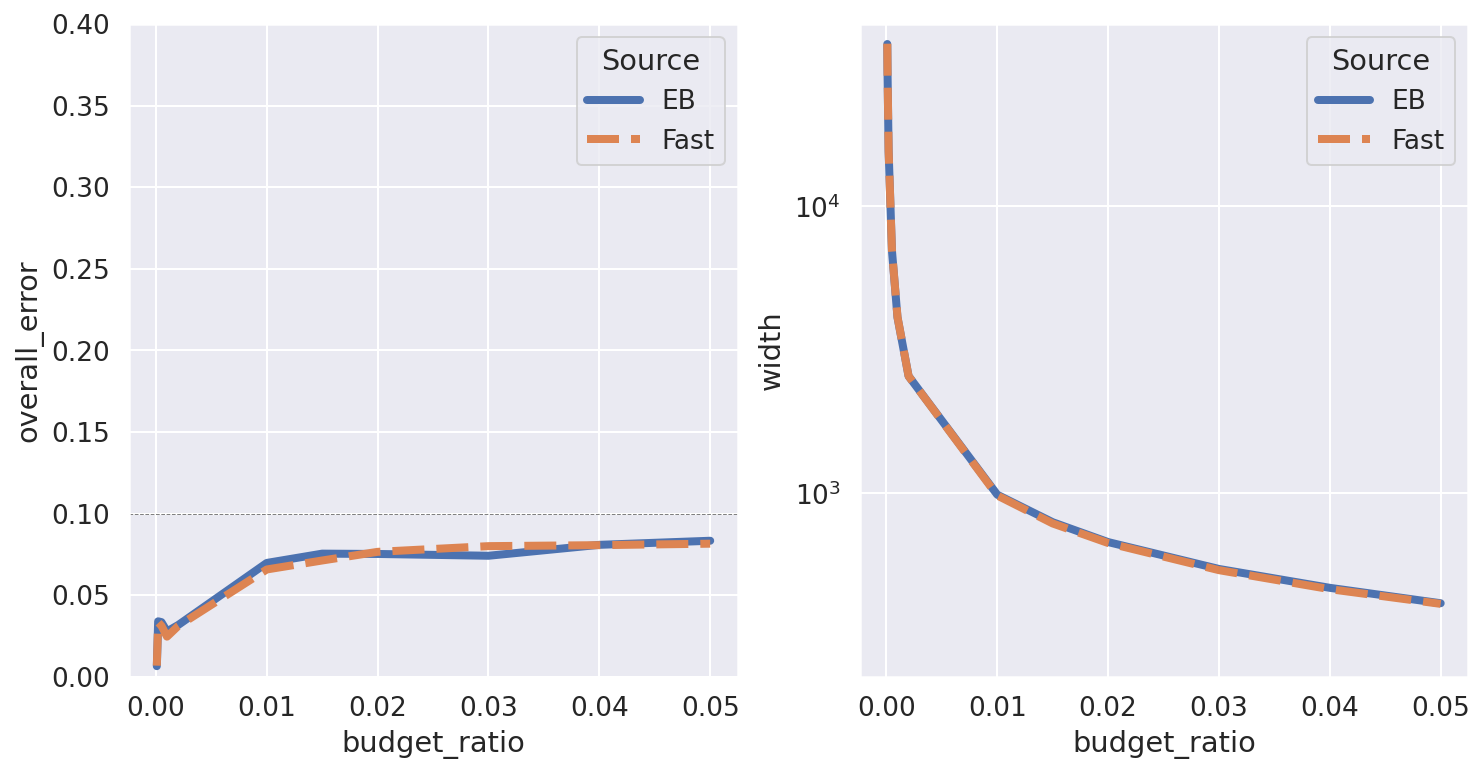}
         
     \end{subfigure}
          \begin{subfigure}[b]{0.8\textwidth}
         \centering
         \caption{$\gamma=0.9$}
         \includegraphics[width=\textwidth]{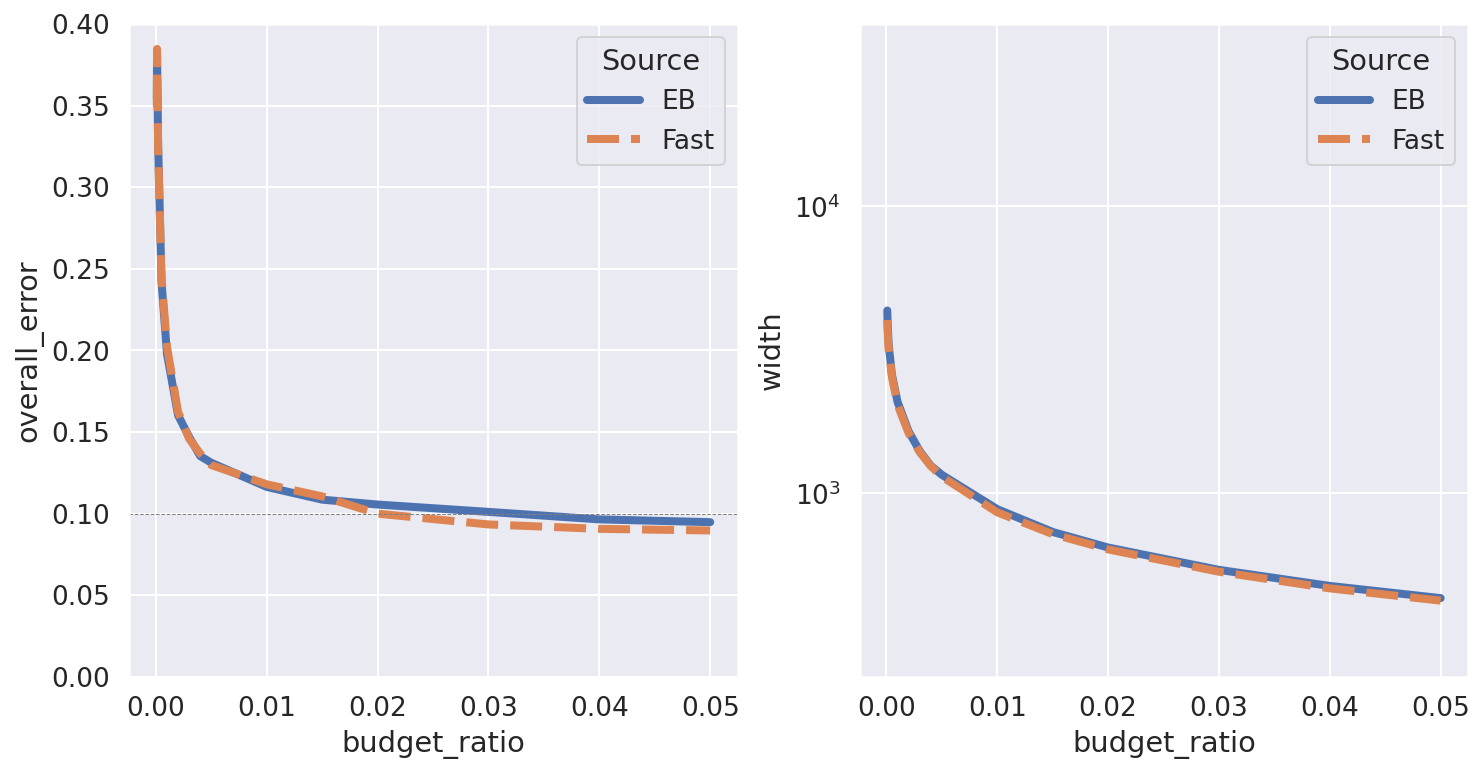}
         
     \end{subfigure}
\caption{Comparison of EB2m (labeled as EB) and the saddlepoint approximation (labeled as Fast) with $\hat{\gamma}=0.5$ in terms of empirical coverage error (left panels) and average CI width (right panels, in log scale) of $90\%$ two-sided CIs with varying $\gamma$ values, where the budget ratio ranges from 0 to 5\%, same setting as Figure 6.3 
in the paper.} \label{fig:eb_vs_saddle2}
     \end{figure}

\subsection{Simulation for two-stage importance sampling}

To simulate two-stage importance sampling, the steps are similar to (1)-(4) as described in Section 6, 
except that step (2) is replaced by (2') below, where the same population configuration i.e. $(\lambda, \pi, f_1, f_0)$ (which also determines $r$) as described above is used:

\begin{enumerate}
\item[(2')] Let $I_{i}\sim Bernoulli(s(V_{i}))$ indicate whether candidate-$i$ is sampled for the first stage for $1\leq i\leq N$, where $s(V_i)=\min(1, Nb_1\times \overline{s}_i)$ with  $\overline{s}_i \propto r(V_i)^{\gamma_1}$ s.t. $\sum_i\overline{s}_i=1$. For candidates sampled for the first stage, i.e. $I_i=1$, let $J_{i}\sim Bernoulli(h(V_{i}))$  indicate whether it is sampled for the second stage, where  $h(V_i)= \min(1, N_1b_2\times \overline{h}_i)$  with $N_1=\sum_iI_i$ and $\overline{h}_i \propto r(V_i)^{\gamma_2}$ s.t. $\sum_i \overline{h}_i=1$. Let $B_i=I_iJ_i$, then $B_i\sim Bernoulli(p(V_i))$ with $p(V_i)=s(V_i)h(V_i)$. So roughly $p(v)\propto r(v)^{\gamma}$ with $\gamma=\gamma_1+\gamma_2$.
\end{enumerate}

The parameters $b_1, b_2$ are the budget ratios for the first-stage and second-stage respectively. We set $b_1=0.1$ and let $b_2$ range from 0 to 0.5 so that the overall budget ratio $b_1b_2$ is from 0 to 0.05, similar to earlier simulations. The parameters $\gamma_1, \gamma_2$ decide how defensive or greedy the two-stage sampling is. 
Figure \ref{fig:two_stage_025_025_05} reports the performance comparison between PB, GO2m, GP2m, EB2 and EB2m with $\gamma_1=\gamma_2=0.25$. 
The results show similar trends of performances to the one-stage importance sampling. The coverage error for PB goes above 50\% as the sampling budget gets close to 0, while EB2m, GO2m and GP2m maintain the nominal coverage for most of the range. 

\begin{figure}
\includegraphics[width=0.8\textwidth]{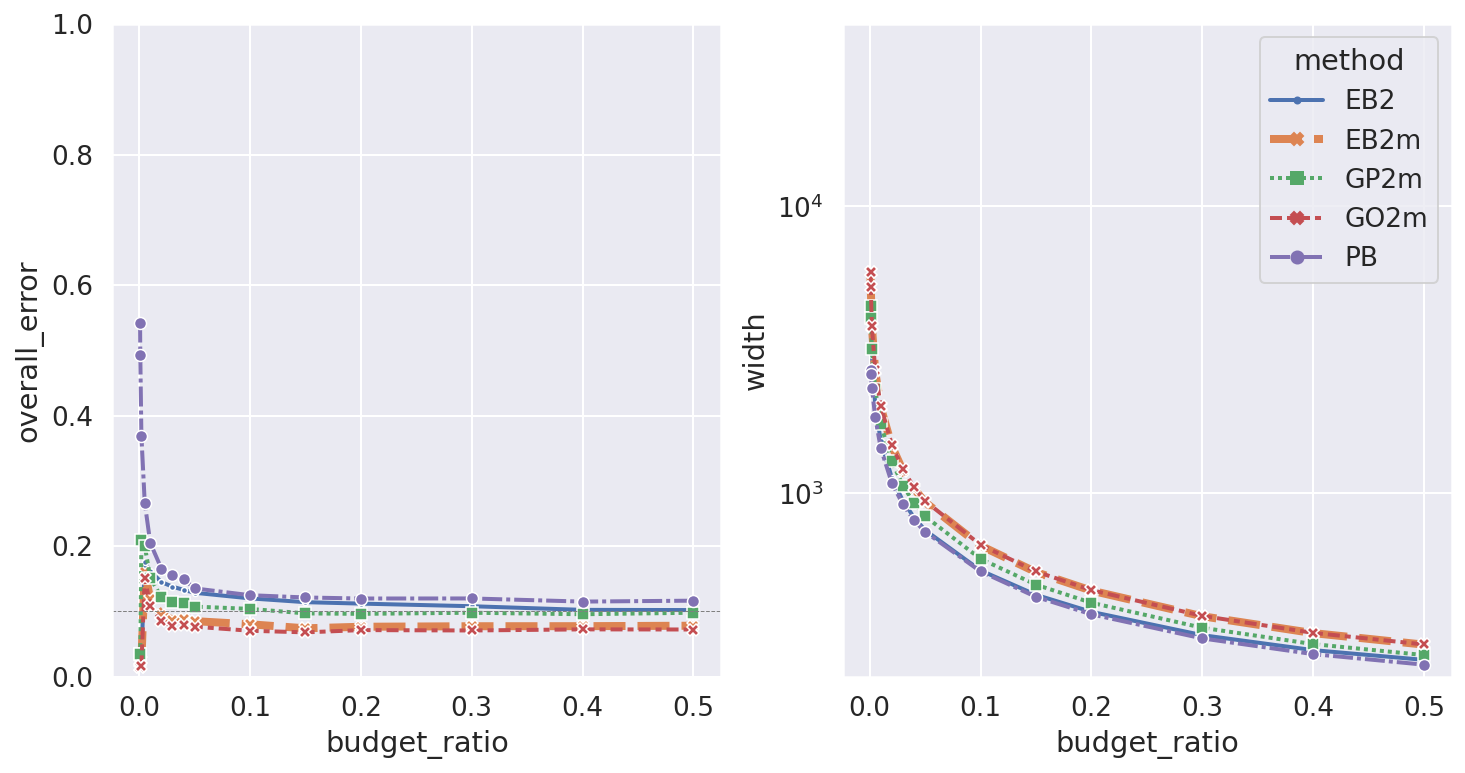} \centering
\caption{ Comparison between PB, GO2m, GP2m, EB2 and EB2m in terms of empirical coverage error
(left panel) and average CI width (in log scale, right panel) of $90\%$ two-sided
confidence intervals in the two-stage importance sampling setup with $\gamma_1=\gamma_2=0.25$, where the x-axis shows the budget ratio for the second stage, which is set to 0.1 for the first stage.}\label{fig:two_stage_025_025_05}
\end{figure}

\label{subsec:multi}

\subsection{Weights for the real-world data analysis}
Here is the data used for the case study reported in Section 7, where all numbers are rounded to two decimal places. Category B consists of a single event with weight 384.69, and Category A consists of 38 true positive events with weights:
[1.    1.    1.    1.    1.    1.    1.    1.    1.    1.    1.    1.
  1.03  1.18  1.18  1.18  1.35  1.38  1.43  1.59  1.72  1.85  1.88  2.09
 11.24 11.24 11.24 11.24 11.25 11.58 12.11 14.39 14.94 15.71 16.1  19.79
 20.   20.]. The estimated value for $||w||_2$ is 72.75 according to (5.7), 
 with $\gamma=0.9$ fitted with some historical data.
\end{appendix}

\end{document}